\documentclass[11pt,a4paper]{article}
\usepackage{mathrsfs}
\usepackage{amsfonts}
\usepackage{amssymb}
\usepackage{bbm}
\usepackage{amsfonts,amssymb, mathrsfs, amsmath,amssymb,float}
\usepackage{amsthm}
\usepackage{amsmath}
\usepackage{algorithm,algorithmic}
\usepackage{blkarray}
\usepackage{color}

\usepackage{geometry}
\newtheorem{myDef}{Definition}[section]
\newtheorem{defn}[myDef]{Definition}
\newtheorem{prop}[myDef]{Proposition}
\newtheorem{exmp}[myDef]{Example}

\newtheorem{lem}[myDef]{Lemma}
\newtheorem{rem}[myDef]{Remark}
\newtheorem{thm}[myDef]{Theorem}
\newtheorem{cor}[myDef]{Corollary}
\newtheorem{alg}[myDef]{Algorithm}
\newtheorem{prob}[myDef]{Problem}

\def\X{{\mathbb{X}}}
\def\Y{{\mathbb{Y}}}
\def\U{{\mathbb{U}}}
\def\V{{\mathbb{V}}}
\def\H{{\mathbb{H}}}
\def\FB{{\mathbb{F}}}
\def\NS{{\mathbb{S}}}

\def\C{{\mathcal{C}}}

\def\m{{\mathfrak{m}}}

\def\deg{\hbox{\rm{deg}}}
\def\sdeg{\hbox{\rm{Sdeg}}}

\def\b{{\bf b}}

\def\Z{{\mathbb{Z}}}

\def\N{{\mathbb{N}}}
\def\R{{\mathbb{R}}}
\def\C{{\mathbb{C}}}
\def\F{{\mathbb{F}}}

\def\i{\mbox{\bf{i}}}

\def\lc{\hbox{\rm{lc}}}
\def\lm{\hbox{\rm{lm}}}

\def\GB{{\mathbb{G}}}

\def\csdeg{\hbox{\rm{CSdeg}}}

\usepackage{environ}
\NewEnviron{iproof}[1][Proof]{\begin{proof}[\indent\bfseries #1]\BODY\end{proof}}{}

\def\mv{{\mathbf{m}}}

\def\FS{{\mathcal{F}}}
\def\HS{{\mathcal{H}}}
\def\MD{{\mathscr{M}}}

\def\bref#1{(\ref{#1})}
\def\proj{\hbox{\rm{Proj}}}

\begin{document}
 %\conferenceinfo{} {}
 %\CopyrightYear{2017}
 %\crdata{978-1-4503-0675-1/17/06}
 %\clubpenalty=10000
 %\widowpenalty = 10000

%\thispagesyle{empty}
\title{\bf Quantum Algorithms for Boolean Equation Solving and Quantum Algebraic Attack on Cryptosystems\thanks{\quad Partially supported by NSFC grants  no. 11688101, no. 11101411.}}
\author{Yu-Ao Chen$^{1,2}$ and Xiao-Shan Gao$^{1,2}$\\
$^{1}$KLMM, Academy of Mathematics and Systems Science\\
 Chinese Academy of Sciences, Beijing 100190, China\\
$^{2}$University of Chinese Academy of Sciences, Beijing 100049, China\\
Email: xgao@mmrc.iss.ac.cn}

\date{\today}
\maketitle
\begin{abstract}
\noindent
Decision of whether a Boolean equation system has a solution is an NPC problem
and finding a solution is NP hard.
In this paper, we present a quantum algorithm to decide whether a Boolean equation system $\FS$ has a solution and to compute one if $\FS$ does have solutions with any given success probability. The runtime complexity of the algorithm is polynomial in the size of $\FS$ and the condition number of $\FS$. As a consequence, we give a polynomial-time quantum algorithm
for solving Boolean equation systems if their condition numbers are small, say polynomial in the size of $\FS$.
We apply our quantum algorithm for solving Boolean equations to the cryptanalysis of several important cryptosystems:
the stream cipher Trivum,
the block cipher AES,
the hash function SHA-3/Keccak,
and the multivariate public key cryptosystems,
and show that they are secure under quantum algebraic attack
only if the condition numbers of
the corresponding equation systems are large.
This leads to a new criterion
for designing cryptosystems that can against the attack
of quantum computers: their corresponding equation systems must
have large condition numbers.



\end{abstract}

\section{Introduction}

As a major advance in quantum algorithm,  Harrow, Hassidim, and Lloyd (HHL) proposed 
a quantum algorithm to solve a linear system $Ax= |b\rangle$ with complexity polynomial in $\log N$, $s$, $\kappa$, and $1/\epsilon$, where 
$N$, $s$, $\kappa$ are respectively the dimension, the sparseness, the condition number of $A$,
and $\epsilon$ is the precision of the output.
The HHL algorithm can be exponentially faster than classic algorithms if $s$ and $\kappa$ are small.
Ambainis gave a new version of the HHL algorithm whose complexity is optimal in $\kappa$ \cite{hhl-new}.
Childs, Kothari, and Somma gave a new quantum algorithm for
linear equation solving, which 
exponentially improves the dependence on the precision \cite{hhl-ha3}.

In this paper, based on the HHL algorithm, a quantum algorithm for Boolean equation solving will be given with complexity polynomial in the size of the input
equation system and the condition number of certain matrix derived from
the equation system.
Solving Boolean equations is a fundamental problem in theoretical computer science.
Decision of whether a Boolean equation system has a solution is an NPC problem and finding a solution is NP hard.
Our algorithm can be as much as exponentially faster than
traditional algorithms for these NP hard problems under certain conditions.

\subsection{Main results}
Let $\FS=\{f_1,\ldots,f_r\}$ be a set of Boolean polynomials in
variables $\X=\{x_1,\ldots,x_n\}$ and with total sparseness $T_\FS = \sum_{i=1}^r \#f_i$, where $\#f_i$ is the number of terms in $f_i$. Then, we have

\begin{thm}\label{th-m1}
For $\epsilon\in(0,1)$, there is a quantum algorithm which decides whether $\FS=0$
has a solution and computes one if $\FS=0$ does have solutions,
with probability at least $1-\epsilon$
and runtime (gate) complexity $\widetilde O((n^{3.5}+T_\FS^{3.5})\kappa^2\log1/\epsilon)$,
where  $\widetilde O$ suppresses more slowly-growing logarithm terms and $\kappa$ is the condition number of the Boolean polynomial system $\FS$ (refer to Theorem \ref{th-q} for definition).
%a maximal condition number for all derived Macaulay matrixes (refer to section ** for definition) of $\FS$.
\end{thm}

As a consequence, we can solve Boolean equation systems using quantum
computers with any given success probability and in polynomial-time
if the condition number $\kappa$ of $\FS$ and the sparseness $T_\FS$ of $\FS$ are small, say when $\kappa$ and $T_\FS$ are poly$(n)$.
Since $T_\FS$ is the size of the input to the algorithm,
it should be small for practical problems.
For instance, all the equation systems  from cryptanalysis
in Section \ref{sec-ca} are very sparse.
Therefore, the key factor is the condition number.
As a consequence, we give a polynomial-time quantum algorithm
for solving Boolean equation systems if their condition numbers are small, say the condition numbers are poly$(n,T_\FS)$.

Let $\FS=\{f_1,\ldots,f_r\}\subset\C[\X]$ be a set of polynomials with complex numbers as coefficients and with total sparseness $T_\FS= \sum_{i=1}^r \#f_i$.
A solution $\mathbf a$ for $\FS=0$ is called {\em Boolean}, if each coordinate of $\mathbf a$ is $0$ or $1$.
Clearly, deciding whether $\FS$ has a Boolean solution is NPC.
We also give a quantum algorithm to compute Boolean solutions of $\FS$.
\begin{thm}\label{th-m2}
For $\epsilon\in(0,1)$,
there is a quantum algorithm which decides whether $\FS=0$
has a Boolean solution and computes one if $\FS=0$ does have Boolean solutions,
with probability at least $1-\epsilon$
and runtime complexity
$\widetilde O(n^{2.5}(n+T_\FS)\kappa^2\log1/\epsilon)$,
where $\kappa$ is the condition number of the polynomial system $\FS$ (refer to Theorem \ref{th-eq} for definition).
%a maximal condition number for all derived Macaulay matrixes (refer to section ** for definition) of $\FS$.
\end{thm}

Theorem \ref{th-m2} is a more basic result.
Theorem \ref{th-m1} follows from Theorem \ref{th-m1} using
a novel reduction, and much more problems such
as optimization over finite fields \cite{qa-opt} can also be
efficiently reduced to Theorem \ref{th-m2}.

We apply Theorem \ref{th-m1} to cryptanalysis of several important cryptosystems.
As early as in 1946, Shannon \cite{sha} pointed out insightfully that
``Construct our cipher in such a way that breaking it is equivalent
to solving a certain system of simultaneous equations
in a large number of unknowns."
We know that the analysis of many cryptosystems,
such as
the stream cipher  Trivum,
the block cipher  AES,
the hash function SHA-3/Keccak,
and the multivariate public key cryptosystems (MPKC),
can be reduced to solving Boolean equations.
Note that all these cryptosystems are important.
AES  is an NIST standard since 2001 \cite{aes-or}, Trivium is an international standard under ISO/IEC 29192-3, and Keccak \cite{kec} is the latest member of the Secure Hash Algorithm family of standards, released by NIST in 2015.

\begin{table}[ht]\centering
\begin{tabular}{|c|c|c|c|c|c|c|c|}\hline
Cryptosystems &${N_k}$&${N_r}$&\#Vars&\#Eqs& $T$ & Complexity \\ \hline
AES-128&4&10&4288&10616&252288&$2^{73.30}c\kappa^2$ \\
AES-192&6&12&7488&18096&421248&$2^{76.59}c\kappa^2$ \\
AES-256&8&14&11904&29520&696384&$2^{78.53}c\kappa^2$ \\
Trivium & &1152&3543&4407&24339&$2^{61.71}c\kappa^2$ \\
Trivium & &2304&6999&9015&49683&$2^{65.38}c\kappa^2$ \\ \hline
 &${N_h}$&${N_r}$&\#Vars&\#Eqs& $T$ & Complexity \\ \hline
Keccak &384&24&76800&77160&611023&$2^{78.25}c\kappa^2$ \\
Keccak &512&24&76800&77288&611540&$2^{78.25}c\kappa^2$ \\ \hline
\end{tabular}
\caption{Complexities of the quantum algebraic attack }
\label{tab-0}
\end{table}

In Table \ref{tab-0}, we give the complexities of using Theorem \ref{th-m1}
to perform quantum algebraic attack to these cryptosystems,
where $\kappa$ is the condition number of the corresponding Boolean equation systems, $T$ is the total sparseness of the Boolean equations,  and $c$ is the complexity constant of the
Harrow-Hassidim-Lloyd (HHL) algorithm (see Remark \ref{rem-HHL} for definition).
For AES-$m$, $m=32N_k$ is the key bit-length  and ${N_r}$ is the number of rounds.
For Trivium,  ${N_r}$ is the number of rounds.
For Keccak, $N_h$ is the output size,  $N_r$ is the number of rounds,
and the state bit-size $b$ is $1600$.
%
%Applying Theorem \ref{th-m1} to these cryptosystems,
From Table \ref{tab-0}, we can see that these cryptosystems are secure under quantum algebraic attack
only if the condition numbers of
their corresponding equation systems are large.
This leads to a new criterion
for designing cryptosystems that can against the attack
of quantum computers: {\em their corresponding equation systems must
have large condition numbers}.
Condition numbers for equation systems are generally difficult to estimate,
and estimating the condition numbers for these cryptosystems
is an interesting future work.

Many problems from computational theory and cryptography can be reduced to finding a Boolean solution for certain polynomial systems.
In this paper, we use Theorem \ref{th-m2} to three such problems:
the 3-SAT problem, the  subset sum problem, and the graph isomorphism problem.

\subsection{Technical contribution}

The main idea of the quantum algorithm proposed in this paper
is that the solutions of a Boolean equation system can be obtained by solving a linear system with the HHL quantum algorithm \cite{hhl}.
For a linear system $Ax= |b\rangle$, the HHL algorithm
can obtain an approximation to the solution state $|x\rangle$ exponentially faster than classic algorithms under certain conditions.
Precisely, our algorithm has three main steps:
\begin{description}
\item[Step 1]
Let $\FS\subset\C[\X]$ have a finite number of solutions:
$\mathbf a_1, \ldots, \mathbf a_w$.
A {\em pseudo solution} of $\FS$ is defined to be a linear combination of {\em monomial solutions} of $\FS$, that is, $\sum_{i=1}^w c_i {\widetilde{\mv}} (\mathbf a_i)$, where ${\widetilde{\mv}}$ is a vector of monomials in $\X$ up to certain degree and $c_i$ are complex numbers.
We show that a pseudo solution of $\FS$ can be computed
by solving a linear system with the HHL algorithm (see section \ref{sec-ps}).
\item[Step 2]
For $\FS\subset\C[\X]$, we show that Boolean solutions for $\FS$ can be computed from the pseudo solutions of $\FS$ obtained in Step 1 with high probability (see section \ref{sec-bs}).
\item[Step 3]
The problem of solving a Boolean equation system is reduced to the computation of the Boolean solutions for a $6$-sparse
polynomial system over $\C$ (see section \ref{sec-be}).
A polynomial system $\HS$ is called $k$-sparse if each polynomial in $\HS$ contains at most $k$ terms.
\end{description}
We will explain each of these steps briefly.

First, we show how to compute the pseudo solutions
for a polynomial system.
Let $\FS=\{f_1,\ldots,f_r\}\subset\C[\X]$ with $d_i=\deg(f_i)$,
$T_\FS = \sum_{i=1}^r \#f_i$, and $D$ a positive integer greater than $\max_i d_i$.
Consider all the polynomials $m_j f_i$, where $m_j$
are monomials with $\deg(m_j) \le D-d_i$.
These equations $m_j f_i=0$ can be written as a linear system
$\MD_{\FS,D} \m_D=\b_{\FS,D}$, where $\m_D$ is the set of all the monomials
$m$  with  $\deg(m) \le D$  and $\b_{\FS,D}$ is the set of the constant terms in $m_j f_i$.
The linear system $\MD_{\FS,D} \m_D=\b_{\FS,D}$ is called
the {\em Macaulay linear system}\footnote{The Macaulay linear system in Section \ref{sec-ps1} is more complicated than the one given here, although they are essentially the same. Here, we use this simple version to explain the ideas.}
and $\MD_{\FS,D}$ is called the {\em Macaulay matrix} of $\FS$.
Our contributions here are in three aspects.
We propose the concept of complete solving degree and show that
for such a degree $D$, the monomial solutions for $\FS$ can be
obtained from $\MD_{\FS,D} \m_D=\b_{\FS,D}$.
We also give a nice upper bound for complete solving degree
of the polynomial system occurred in computing the Boolean solutions.
It is shown that by using the HHL algorithm to $\MD_{\FS,D} \m_D=\b_{\FS,D}$,
we obtain a pseudo solution for $\FS$.
We gave a modified HHL algorithm to solve the
Macaulay linear system, which has better complexities
than using the original HHL algorithm.
%
%Directly using the HHL algorithm to the Macaulay linear system,
%we will have an undesirable complexity bound.
%We show that $\MD_{\FS,D}$ is $T$-sparse and can be written
%as the summation of $T$  matrices of $1$-sparseness, which
%can be quarried efficiently. Based on these properties of
%$\MD_{\FS,D}$, we gave a modified HHL algorithm to solve
%$\MD_{\FS,D} \m_D=\b_{\FS,D}$ with better complexity.
%%
%Our contribution here are two folds. We show that
%$\MD_{\FS,D}$ is a $T$-sparse matrix, which allows
%us to use the HHL algorithm.
%We also give a formula for the solution of $\MD_{\FS,D} \m_D=\b_{\FS,D}$
%in terms of the solutions of $\FS$
%when using the HHL algorithm \cite{hhl} to it,
%which is called a {\em monomial solutions} of $\FS=0$.

Second, we show how to compute Boolean solutions
for a polynomial system $\FS$ over $\C$,
which are the solutions of $\FS_1=\FS\cup \{x_1^2-x_1,\ldots,x_n^2-x_n\}$
over the field of complex numbers.
Our contribution here is to show that
the solutions of $\FS_1=0$ can be obtained
from the  pseudo solutions of $\FS_1=0$ with high probability
by combining the property of quantum states and that of Boolean solutions.
The novelty of the approach is that the error bound for the solutions in the HHL algorithm is used to give the probability for finding
Boolean solutions of $\FS$.

Thirdly, let $\FS$ be a Boolean polynomial system in variables $\X$.
Since the HHL algorithm works over $\C$ and does not work for finite fields,
we cannot use the HHL algorithm to the Macaulay linear system of $\FS$.
We prove that the solutions to $\FS$ are the same
as the Boolean solutions of a $6$-sparse polynomial system
$\FS_2\subset\C[\X,\U]$ for some extra indeterminates $\U$.
Furthermore, the numbers of variables in $\U$ and the numbers of equations in $\FS_2$ are linear in the size of $\FS$.
By computing the Boolean solutions of $\FS_2$, we find the
solutions of $\FS=0$.

\subsection{Comparing with existing work}

The idea of reducing nonlinear polynomial systems to
linear systems of monomials can be traced back to the classical 
work on resultants \cite{mac}, which give conditions
for the existence of common solutions for over-determined polynomial systems.
Here, a key concept is the solving degree. Precisely,
$D$ is a solving degree of polynomial system $\FS$,
if the Gr\"obner basis of $(\FS)$ can be obtained  by using Gaussian elimination to
the Macaulay linear system $\MD_{\FS,D} \m_D=\b_{\FS,D}$.
Lazard \cite{laz83} and more recently
Caminata-Gorla \cite{soldeg-1} gave nice upper bounds for the solving degree
for projective zero dimensional polynomial ideals.
The F4 algorithm \cite{f4} and the XL \cite{xl}
algorithm were proposed to compute the Gr\"obner basis based on this idea.
Since the Macaulay linear system for a polynomial system with
multiple roots are under-determined, it is not possible
to solve the polynomial system by computing values of the monomials directly,
unless the equation system has a unique solution.
In the general case, extra work need to be done,
such as to reduce polynomial system solving to
the computation of eigenvalues or solving of univariate equations  \cite[p51]{cox2}.
% or  the rational univariate representation \cite{rur}.
%

In this paper, we show that the solving degree is not big enough
for monomial solving from the Macualay linear system.
Instead, we propose the concept of complete solving degree
which is enough to for monomial solving from the Macaulay linear system.
Furthermore, we show that the Boolean solutions
of a polynomial system can be found directly
from the solutions of the Macaulay linear system
in the quantum case.

Our algorithm is based on the HHL algorithm
for solving the linear system $Ax=|b\rangle$, where $A\in\C^{N\times N}$, $x,b\in\C^N$.
The speedup achieved in our algorithm is based on the
exponential speed up of the HHL algorithm for solving sparse linear systems.
On the other hand, the HHL algorithm has the following subtle properties.
\begin{enumerate}
\item The algorithm does not give a solution to $Ax= |b\rangle$, but a state $|{x}\rangle=(x_1,\ldots,x_N)$. Measuring of $|{x}\rangle$  gives $|x_1|:|x_2|:\cdots:|x_N|$ and the complexity will increase to $O(N)$.

\item The algorithm gives {an} answer $|x\rangle$ even if $A|x\rangle= |b\rangle$ has no solutions.

\item The algorithm works over $\C$, but not over finite fields.
%
%\item The algorithm gives an approximation to the state
%$|\widehat{x}\rangle$ with any error bound $\nu \in (0,1)$.
%and the complexity of the algorithm has a factor $1/\nu$.
%
%
%\item It is supposed that $e^{iAt}$ can be obtained efficiently.
%\item The algorithm assumes that the state $|b\rangle$ can be generated
%efficiently.
\end{enumerate}
Our algorithm does not have these limitations
and gives an exact solution to the Boolean system.
For instance, finding the Boolean solutions to a linear equation
with integer coefficients is an NP hard problem (Section \ref{sec-lin}),
which can be solved by our algorithm but not the HHL algorithm. 
So our algorithm can be considered a significant extension of the HHL algorithm.
%
%Also, our algorithm is a probabilistic algorithm with success probability
%$\ge 1-\epsilon$ and the complexity of algorithm has a factor $\log (1/\epsilon)$ for any $\epsilon\in(0,1)$.
%

It is interesting to see that the second ``drawbacks" of the HHL algorithm mentioned above is used to generate the quantum state $|b\rangle$ efficiently (see Lemma \ref{lm-hhl3}) for the Macaulay linear system.
The HHL algorithm assumes that $|b\rangle$ is given. In the general case, there exist no efficient algorithms to generate $|b\rangle$ from $b$ \cite{hhl-ha1} and efficient generation for $|b\rangle$ can be achieved only in some special cases \cite{hhl-ci}.
Fortunately,  in our case, $b=\b_{\FS,D}$ is very sparse: at most $r$ entries of $b$ are nonzero, which leads to an efficient generation for $|b\rangle$ and the complexity is negligible
comparing to that of the HHL algorithm.

The complexity of our algorithm contains the condition number
of the Boolean polynomial system, which is the condition numbers 
of the Macaulay matrixes of the polynomial systems 
and is inherited from the HHL algorithm.
%
%The limitation on the condition number is inherited
%from the quantum algorithm HHL for solving linear equations \cite{hhl},
It was proved in \cite{hhl} that the dependence on condition
number cannot be substantially improved.
More precisely, it was shown that if a quantum algorithm exists for matrix inversion
running in time $O(\kappa^{1-\delta}\hbox{poly}(\log(N))$ for some
$\delta>0$, then {\bf BQP} ={\bf PSPACE} \cite{hhl}.
%
%Also note that, for a symmetric and positive-definite $A\in\C^{N\times N}$,
%the best classic numerical method
%for solving the linear equation $Ax=b$ has complexity $\widetilde{O}(N\sqrt{\kappa})$ which also depends on the condition number  $\kappa$ of $A$ \cite{c19}.
%When $\kappa$ is large, $A$ is called {\em ill-positioned}.
%
%But, the condition number in our case is much more complicated.
%
Although the condition numbers of generic polynomial systems are usually exponential, the polynomial systems to be solved in our problems
(refer to sections \ref{sec-lin}, \ref{sec-3sat}, \ref{sec-ca}) 
are highly structural, which  are super sparse in most cases
and always include the ``Boolean equations" $x_1^2-x_1,\ldots, x_n^2-x_n$. 
It is expected that new methods are needed to estimate the condition numbers for these polynomial systems.
%
%
%the HHL algorithm needs $O(\log(N))$-qubit registers \cite{hhl}
%due to the inherent advantage of quantum computers:  an $N$-dimensional vector is represented by a $O(\log_2 N)$-qubit state.
%In our case, the dimension of $\m_D$ is  $O(e^{n+T})$ (see Lemma \ref{lm-qde}), as a consequence, our algorithm needs $O(n+T)$-qubit registers, where $n$ is the number of variables and $T$ is the total sparseness of the equation system.
%
%Finally, regarding to the HHL algorithm,
%it is pointed out in \cite{childs} that
%``Will the quantum solution of linear
%equations turn out to be a widely used
%tool, or are its limitations too great for the
%technique to be of practical significance?"
%It seems that the result of this paper does provide
%a significant application of the HHL algorithm.

Previous work on quantum polynomial system solving and algebraic
cryptanalysis were mainly based on Grover's algorithm \cite{grover}
which can achieve quadratic speed-up for exhaust search.
The idea of these work are similar: reduce
the problem to be solved to a search problem by designing a
proper oracle and then use Grover's algorithm.
Schwabe-Westerbaan \cite{qmq1} and more recently Faug\`ere et al \cite{qmq2}
proposed quantum algorithms for solving Boolean multivariate quadratic polynomial systems (MQ).
In particular, a Las-Vegas quantum algorithm for solving Boolean MQ
with complexities $O(2^{0.462n})$ was given under certain conditions.
On the other hand, the complexity of our algorithm is polynomial
in $n$, the total sparseness, and the condition number.
One nice feature of our algorithm is that the complexity depends
on the structure of polynomial system and provides
faster algorithms for sparse systems with small condition numbers.

In \cite{aesq1}, Grassl et al presented quantum circuits to search the
key for AES based on Grover's algorithm. This is possible because,
once the keys are given, the polynomial system for AES
can be solved easily. For AES-128, 192, and 256, the gates
(complexity) used in the circuits are about
$2^{86}$, $2^{118}$, $2^{151}$ \cite{qmq1}, respectively.
Comparing to our result in Table \ref{tab-0}, we have the following observations.
The ratio of the complexities of AES-256 to that of AES-128
in Table \ref{tab-0} is $2^5$, while the same ratio for the complexities in \cite{aesq1} is $2^{65}$.
This shows the advantage of the polynomial-time
nature of our algorithm: when the key is doubled,
the complexity increases little, while
for the exponential algorithm in \cite{aesq1},
when the key doubles, the complexity increases exponentially.
On the other hand, using Grover's algorithm, the explicit complexities
were given in \cite{qmq1}, while
our algorithm is much complicated and the complexities in Table \ref{tab-0} contain the parameters $c$ and $k$.

\section{A modified HHL algorithm}

In this section, we give a modified HHL algorithm for solving the linear system $Ax=b$, where special assumptions about $A$ and $b$
are made. The modified HHL algorithm will be used
in our algorithm for solving Boolean equations.

For a matrix $A\in\C^{M\times N}$, the arithmetic square root of each nonzero eigenvalue of the matrix $A^\dagger A$ is called a {\em singular value} of $A$, and the quotient of the maximal and minimal singular values is called the {\em condition number} of $A$, where $A^\dagger$ denotes the complex conjugate transpose of $A$.
A matrix $A\in\C^{M\times N}$ is called {\em $s$-sparse} if each row and column of $A$ has at most $s$ nonzero entries.
We say that the {\em  complexity of a query} for $A$ is $O(\gamma)$, if there is an algorithm to find all the nonzero entries in each row or column of $A$ in time $O(\gamma)$.

The following HHL quantum algorithm \cite{hhl} was given to solve a linear equation system $A|x\rangle= |b\rangle$ over $\C$.
\begin{thm}\label{thm-hhl0}
Suppose that $A\in\C^{M\times N}$ is an $s$-sparse matrix
and the complexity of a query for $A$ is $O(s)$.
Let  $|b\rangle\in\C^M$ be a unitary quantum state.
Then, there is a quantum algorithm   which can give an $\epsilon$-approximation to a solution state of the linear system $A|x\rangle=|b\rangle$  in runtime complexity $\widetilde O(\log(M+N)s^2\kappa^2/\epsilon)$,
where $\kappa$ is the condition number of $A$.
\end{thm}

As usual, the notation $\widetilde O$ suppresses more slowly-growing logarithm terms.
With the best known algorithm for Hamiltonian simulation \cite{he},
the complexity of the HHL algorithm can be reduced to $\widetilde O(\log(M+N)s\kappa^2/\epsilon)$.
Ambainis gave a new version of the HHL algorithm which has complexity $\widetilde O(\log(M+N)s^2\kappa/\epsilon^3)$ \cite{hhl-new}.
Childs-Kothari-Somma gave a new quantum algorithm for
linear equation solving, which has complexity
$\widetilde O(s\kappa^2\log(1/\epsilon)(\log(M+N) + \log^{2.5}(1/\epsilon)))$  \cite{hhl-ha3}.

In the rest of this section, we will present
a modified version of the HHL algorithm
under the following assumptions.

{\bf Assumption 1}.
$A\in\C^{M\times N}$ is $s$-sparse and has a decomposition $A=\sum_{j=1}^sA_j$,
where each $A_j$ is a $1$-sparse matrix with the complexity of a query to be $O(\gamma)$.

{\bf Assumption 2}. Let $b\in\{0,1\}^{M}$  and
 $M=r2^v$ for  positive integers $v, r$.
Furthermore,  $b[i]=1$ if and only if
$i=k2^v$ for $k=0,\ldots,\rho-1$,
where $\rho$ is a positive integer $\le r$.

We have the following modified  HHL algorithm which follows from Lemma \ref{thm-hhl1}.
%
%\begin{thm}[\cite{he,hhl,hhl-ha1,hhl-ha2}]\label{HHL}
%
\begin{thm}\label{thm-hhl}
Under {\bf Assumptions} 1 and 2, the HHL algorithm can give an $\epsilon$-approximation to a solution state  of the linear equation system $Ax=b$  in runtime complexity $\widetilde O((\log(M+N)+\gamma)s\kappa^2/\epsilon)$,
where  $\kappa$ is the condition number of $A$.
\end{thm}

\begin{rem}\label{rem-HHL1}
Theorem \ref{thm-hhl} differs from Theorem \ref{thm-hhl0} in the following aspects.
(1) $b$ is given as a vector instead of a state $|b\rangle$.
(2) $A$ satisfies different conditions: $A$ is $s$-sparse, but
 its complexity of a query is $O(s\gamma)$, while the complexity of a query for $A$ is $O(s)$ in Theorem \ref{thm-hhl0}.
(3) $\gamma$ is added to the complexity.
(4) In the complexity, $s^2$ is reduced to $s$, which
can also be done with best known algorithm for Hamiltonian simulation \cite{he}, but under the condition that the complexity of a query for $A$ is $O(s)$.
\end{rem}

\begin{rem}\label{rem-HHL}
In the cryptanalysis to be given later in this paper, we
make the following approximation to the complexities
of the HHL algorithm $\widetilde O((\log(M+N)+\gamma)s\kappa^2/\epsilon)
\simeq c(\log(M+N)+\gamma)s\kappa^2/\epsilon)$, where
$c$ is called the {\em complexity constant} of the HHL algorithm.
\end{rem}

For a matrix $A\in\C^{M\times N}$, denote $I(A)=\left(\begin{array}{cc}
                                                      0 & A \\
                                                      A^\dagger & 0 \\
                                                    \end{array}\right)\in\C^{(N+M)\times (N+M)}$,
which is a Hermitian matrix.
In fact, the HHL algorithm will solve the linear system
%\begin{equation}
$$
\left(\begin{array}{cc}
0 & A \\
A^\dagger & 0 \\
\end{array}\right)
\left(\begin{array}{c}
0\\
x\\
\end{array}\right)
=
\left(\begin{array}{c}
b \\
0\\
\end{array}\right)
$$
%\end{equation}
%and give an $\epsilon$-approximate state to a solution state.
instead of $Ax=b$ \cite{hhl}.

We will prove Theorem \ref{thm-hhl} in two steps: first consider the
Hamiltonian simulation for $e^{\i I(A)t}$ and second consider the preparation for the state $|b\rangle$, where $\i=\sqrt{-1}$.
We need the following result about the quantum complexity for the Hamiltonian simulation.
\begin{lem}[\cite{ha}]\label{lm-ha}
For a $1$-sparseness decomposition $A=\sum_{j=1}^uA_j$ of a matrix $A\in\C^{M\times N}$  and  a given time $t$, we can quantumly simulate
$e^{\i I(A)t}\simeq (\prod_{j=1}^ue^{\i I(A_j)t_0})^{t/t_0}$ for any small number $t_0$ by $O(\log(M+N)(\log^*(M+N))^2ut)=\tilde{O}(\log(M+N)ut)$ auxiliary operations and totally $O(\log^*(M+N)ut)$ queries for the $A_j$s, where $\log^*(n)=\min\{r\,|\,\log_2^{(r)}n<2\}$ (the $^{(r)}$ indicating the iterated logarithm).
\end{lem}

In the following lemma, we modify the HHL algorithm to take into the fact that $A$
already has a $1$-sparseness decomposition.
\begin{lem}\label{HHL}
Under {\bf Assumption} 1, the HHL algorithm gives an $\epsilon$-approximation to a solution state %$|\tilde x\rangle$
of the linear system $Ax=|b\rangle$  in time $\widetilde O((\log(M+N)+\gamma)s\kappa^2/\epsilon)$,
where $\kappa$ is the condition number of $A$.
\end{lem}
\begin{proof}
The complexity for the HHL algorithm comes from the Hamiltonian simulation $e^{\i I(A)t}$ for time $t=\kappa^2/\epsilon$ \cite{hhl}.
It is proved that $I(A)$ can be decomposed as the summation of $(\log^*(M+N) s^2)$ $1$-sparse matrices.
By {Lemma \ref{lm-ha}}, the complexity for the HHL algorithm
is $\widetilde O(\log(M+N)s^2\kappa^2/\epsilon)$, where the complexity
for the query is negligible \cite{hhl}.

Under {\bf Assumption} 1, since $I(A)$ can be decomposed as the summation of $s$ matrices of $1$-sparseness, by {Lemma \ref{lm-ha}}, the complexity for the modified HHL algorithm will decrease to $\widetilde O(\log(M+N)st+\log^*(M+N)st\gamma)|_{t=\kappa^2/\epsilon}=\widetilde O((\log(M+N)+\gamma)s\kappa^2/\epsilon)$, since $\log^*(M+N)$ is approximately a constant \cite{ha}.
%, since $\log^*(M+N)$ is always less than or equal to $6$ \cite{ha}.
\end{proof}

We need the following detailed information about the HHL algorithm.
\begin{lem}\cite{hhl}\label{lem-HHL}
%
%Suppose that $A$ satisfied the condition of Theorem \ref{thm-hhl}.
Let $\lambda_1,\ldots,\lambda_n$ be the singular values  of $A\in\C^{M\times N}$, $|v_j\rangle$ ($|u_j\rangle$)   the eigenvectors of $A^\dagger A$ ($AA^\dagger$) with respect to the nonzero eigenvalues $\lambda_j^2$ of $A^\dagger A$, and thus $A=\sum\limits_{j=1}^n\lambda_j|u_j\rangle\langle v_j|$ is the singular value decomposition of $A$.
%
%Let  $|b\rangle\in\C^M$ be a unitary quantum state.
Then, the HHL algorithm returns an $\epsilon$-approximation to the solution state $|\frac{\tilde x}{\|\tilde x\|}\rangle$, where
\begin{eqnarray}\label{eq-hhl}
\widetilde{x}=\sum\limits_{j=1}^n\lambda_j^{-1}|v_j\rangle\langle v_j|b\rangle.
\end{eqnarray}
Furthermore,  $\tilde x$ has the minimal norm $\|\tilde x\|=\sqrt{\langle \widetilde{x},\widetilde{x}\rangle}$ among all solutions for $Ax=b$.
\end{lem}
%A state $|x\rangle$ is said to be an approximation to another state $|y\rangle$
%with an error bound $\epsilon$ if $\| \frac{x}{\|x\|}-\frac{y}{\|y\|}\| < \epsilon$.
%
%
%Even though the HHL algorithm will solve $I(A)(y,x)^T=(b,0)^T$ instead of $Ax=b$, the special solution with minimal norm is exactly $(0,\tilde x)^T$, which comes from the property of the HHL algorithm. Equivalently, we obtain the special solution $\tilde x$.
%\begin{rem} It is easy to check $(y,\tilde x)^T$ is a solution for $I(A)(y,x)^T=(b,0)^T$ for any $y\in\C^M$, and $(0,\tilde x)^T$ has the minimal norm among all solutions.
%\end{rem}

In Lemma \ref{HHL}, $|b\rangle\in\C^M$ is given as a
quantum state and there exist no efficient algorithms to generate $|b\rangle$ in the general case \cite{hhl-ha1}.
%Since the HHL algorithm will solve $I(A)(y,x)^T=(b,0)^T$ instead of $Ax=b$, we need to prepare state $|(b,0)^T\rangle\in \C^{M+N}$. Without loss of generality, we denote $|(b,0)^T\rangle$ by $|b\rangle$ under {\bf Assumption} 2.
In the rest of this section,
we will modify the HHL algorithm such that the input to the HHL algorithm is $b$ instead of  $|b\rangle$ under {\bf Assumption} 2.
We first prove a lemma.
\begin{lem}\label{lm-hhl3}
Let $Bx = c$ be obtained by adding  more ``equations" ${0}x = 1$ to
$A|x\rangle =|b\rangle $. Then using HHL to $B|x\rangle  = |c\rangle$, we obtain the same solution state as that of $A|x\rangle =|b\rangle $.
\end{lem}
\begin{proof}
Let $B = \left( \begin{array}{c}A \\ 0  \\\end{array} \right)$ and
$c = \left( \begin{array}{c}b \\ 1  \\\end{array} \right)$,
where we use $0$ ($1$) to represent certain maxtix of zeros (ones) with the proper dimension.
We have $B^\dagger B= A^\dagger A$.
Then, adding some $0$ rows to $A$ will not change the nonzero eigenvalues of $A^\dagger A$ and the eigenvectors of $B^\dagger B$
are the same as that of $A^\dagger A$.
Now, the lemma  follows from  Lemma \ref{lem-HHL}.
\end{proof}

In the following lemma, we modify the HHL algorithm to take into the fact that $b$ contains $\rho$ nonzero entries.
%{\begin{rem}
%In this paper, when we count a ordinal number in a matrix or a vector, we start from $0$.
%\end{rem}}

\begin{lem}\label{thm-hhl1}
Under {\bf Assumptions} 1 and 2,
%
%Assume that only the first  $\rho$ entries of $b\in\C^M$ are nonzero
%and $A\in\C^{M\times N}$ is $s$-sparse.
the HHL algorithm can give an $\epsilon$-approximation to a solution state  of the linear system $Ax=b$  in time $\widetilde O((\log(M+N)+\gamma)s\kappa^2/\epsilon)$.
\end{lem}
\begin{proof}
%Let $b=(b_0,b_1,\ldots,b_{r2^\delta-1})$ with only $b_{i2^\delta}=1$ for $i=0,1,\ldots,\rho-1$.
%By dividing $b_{i2^\delta}$ to the $i2^\delta$-row  of $Ax=b$,
%we may assume $b_{i2^\delta}=1$ for $i=0,\ldots,\rho-1$.
%This step costs $O(T_\rho)$.

Let $\sigma =  {\lceil\log_2 \rho\rceil} \le \log_2 r+1$.
Let ${w}\in\{0,1\}^{2^v}$ such that $w[0]=1$ and
$w[i]=0$ for $i=1,\ldots,2^v-1$.
Adding $2^\sigma-\rho$  blocks ${w}$ before the $\rho2^v$-th row of $b$, we have
a vector $c$ such that $c[i]=1$ if and only if $i=k2^v$ for $k=0,1,\ldots,2^\sigma -1$.
Correspondingly, by adding
$(2^\sigma-\rho)2^v$ zero rows to $A$ before the $\rho2^v$-th row of $A$, we obtain
a matrix $B\in\C^{((r+2^\sigma-\rho)2^v)\times N}$.
Note that adding some zero rows will not increase the sparseness and the complexity of a query for a matrix.
Then the equation system $Ax=b$ becomes
\begin{equation}\label{eq-bx1}
B x = c
\end{equation}
We may add more zero rows at the ends of $B$ and $c$ such that $B\in\C^{2^{\eta}\times N}$ and $c\in\C^{2^\eta}$,
where $\eta=\lceil\log_2 (r+2^\sigma-\rho)\rceil+v$.
With these assumptions, we can easily generate the state $|c\rangle$:
$$|c\rangle =
\otimes_{i=1}^{\eta-\sigma-v}|0\rangle \otimes_{i=1}^{\sigma}(H|0\rangle)\otimes_{i=1}^{v}|0\rangle ,$$
where $H=\frac{1}{\sqrt 2}\left(
           \begin{array}{cc}
             1 & 1 \\
             1 & -1 \\
           \end{array}
         \right)
$ is the Hadamard operator.
The complexity of generating $|c\rangle$ is $O(\eta)$ = $O(v+\log r)=O(\log(M))$,
since  $\sigma =  {\lceil\log_2 \rho\rceil} \le \log_2 r+1$ and $M=r2^v$.
The equation system \bref{eq-bx1} becomes
\begin{equation}\label{eq-bx2}
C  |x\rangle = |c\rangle
\end{equation}
where $C=\frac{B}{2^{\sigma/2}}$. We show that $C$ satisfies {\bf Assumption} 1.
It is clear that $C$ can be written as the summation of $s$ $1$-sparse matrices $C_{j}$.
The complexity of a query for $C_{j}$ is the same as that of $A_j$,
because the $(u,v)$-th element of $C$ is
the $(u,v)$-th element of $B$ divided by  $2^{\sigma/2}$,
in other words, we do not need to actually generate the matrix $C$.

By Lemma \ref{lm-hhl3}, equation system \bref{eq-bx2} has the same solution state as that of $A|x\rangle =|b\rangle $, when using the HHL algorithm to them.
As mentioned before, the HHL algorithm actually solves $I(C)(0,x)^T=(c,0)^T$.
Similar to the above procedure, we can add more zeros to the end of
$(c,0)^T$ to obtain a state, which costs at most $O(\log(2^\eta+N))=O(\log(M+N))$.

By Lemma \ref{HHL}, the total complexity is the complexity of using the HHL algorithm to \bref{eq-bx2} plus that of generating $|c\rangle$,
that is, $\widetilde O((\log(M+N)+\gamma)s\kappa^2/\epsilon+\log(M+N)) = \widetilde O((\log(M+N)+\gamma)s\kappa^2/\epsilon)$.
\end{proof}

\begin{rem}
If $\gamma$ is small, say $\gamma = O(\log(M+N))$,
then the complexity of the modified HHL algorithm
is $\widetilde O((\log(M+N)+\gamma)s\kappa^2/\epsilon)
= \widetilde O(\log(M+N)s\kappa^2/\epsilon)$.
%
%In this case, the HHL algorithm could be used to solve the equation system $Ax=b$ in complexity $\widetilde O(\log(M+N)s\kappa^2/\epsilon)$.
Fortunately, the linear system to be solved in Section \ref{sec-ps}  has
this property.
\end{rem}

\section{Quantum monomial-solving of polynomial systems over $\C$}
\label{sec-ps}
In this section, we give a quantum algorithm to find the
a solution for the monomials of a polynomial system $\FS$,
which satisfy a linear system.

\subsection{The Macaulay linear system}
\label{sec-ps1}
In this section, we will construct a Macaulay linear system
for a finite set of polynomials and show that the linear system satisfies
{\bf Assumptions} 1 and 2 given in Section 2.

Let $\C$ be the field of complex numbers and
$\C[\X]$ the polynomial ring in the indeterminates
$\X = \{x_1,\ldots,x_n\}$.
For a polynomial $f\in \C[\X]$, denote $\deg(f)$, $\# f$, and
$\m(f)$ to be the total degree of $f$, the sparseness (the number of terms) of $f$,
and the set of monomials of $f$.
For $S\subset\C[\X]$, we use $\V_{\C}(S)\subset \C^n$ to denote the common zeros of the polynomials in $S$.

Let $\m$ denote the set of all the monomials in variables $\X$.
In this section, we will use the lexicographic monomial ordering for $x_1>\cdots>x_n$. For convenience, we denote $\mathbf 0=(0,0,\ldots,0),\mathbf 1=(1,1,\ldots,1)\in\C^n$, and $\X^\alpha=\prod_{i=1}^nx_i^{\alpha_i}$ for $\alpha=(\alpha_1,\ldots,\alpha_n)\in\Z^n$. For a given positive integer $d$, let $\m_{\le d}$ be the set of all monomials which are factors of $\X^{d\cdot\mathbf 1}=x_1^dx_2^d\cdots x_n^d$. We sort
$$\m_{\le d}=\{m_{d,0},m_{d,1},\cdots,m_{d,(d+1)^n-1}\}$$
in ascending lexicographic monomial ordering. Then $m_{d,0}=1,m_{d,1}=x_n,m_{d,2}=x_n^2,m_{d,d}=x_n^d,
m_{d,d+1}=x_{n-1}$ and $m_{d,(d+1)^n-1}=\X^{d\cdot\mathbf1}$.
Also note that $\#\m_{\le d} = (d+1)^n$.

In the rest of Section 3, let $\FS = \{f_1,\ldots,f_r\}\subset\C[\X]$ with  $d_i=\deg(f_i)$ and $t_i=\# f_i$ for $i=1,\ldots,r$.
%We always assume that $(\FS)$ is zero-dimensional, so $m\ge n$.
%
%Without loss of generality, we assume that each constant term $f_i(\mathbf 0)$ is either $0$ or $-1$ for $i=1,\ldots,r$.
%We will construct a modified Macaulay matrix for $\FS$.
We first define several parameters.
\begin{defn}\label{def-par}
Without loss of generality, we may assume $f_i(0)=-1$ for $i=1,\ldots\rho$
and $f_i(0)=0$ for $i=\rho+1,\ldots,r$.
Let $D\in \N$ such that $D\ge\max_{i=1}^r d_i$.
Let $\bar{d}$ be the minimal integer satisfying $\bar{d}\ge D-\min_id_i$ and $\bar{d}+1 =2^{\delta}$ for certain $\delta\in\N$. Set $\bar D$ to be the minimal integer satisfying $\bar D\ge D$ and $\bar{D}+1=2^{\Delta} $ for certain $\Delta\in\N$.
\end{defn}

{\begin{rem}\label{rem-sb}
In this paper, the subscripts for a matrix or a vector always start from $0$,
because, the complexity analysis of the algorithm in this paper depends
on the representation of the  subscripts.
\end{rem}}

For $i=1,\ldots, r$ and each $m_{\bar d,j}\in\m_{\le\bar d}$ with $\deg(m_{\bar d,j})\le D-d_i$, $m_{\bar d,j}f_i$ could be considered as a linear function in the monomials in $\m_{\le \bar D}$. For $\deg(m_{\bar d,j})> D-d_i$, we replace $m_{\bar d,j}f_i$ by $0$.
In order to write these functions precisely, introduce the following new notation:
\begin{equation}
m_{\bar d,j, i} =\begin{cases}
m_{\bar d,j} ,& \text{if  }  \deg(m_{\bar d,j})\le D-d_i;\\
0,& \text{if } \deg(m_{\bar d,j})> D-d_i
\end{cases}
\end{equation}
for $i=1,\ldots, r$ and $j=0,\ldots,(\bar {d}+1)^n-1$.
We rewrite $m_{\bar d,j,i}f_i$ for $i=1,\ldots,r$ and $j=1,\ldots,({\bar {d}}+1)^n-1$ in matrix form:
\begin{eqnarray}\label{eq-1}
\begin{blockarray}{cc}
&m_{\bar D,1}<m_{\bar D,2}<\cdots<m_{\bar D,({\bar D}
+1)^n-1}\cr
\begin{block}{c(c)}
m_{\bar {d},0,1}f_1&\cdots\cr
\vdots&\cdots\cr
m_{\bar {d},{(\bar {d}}+1)^n-1,1} f_1&\cdots\cr
m_{\bar {d},0,2}f_2&\cdots\cr
\vdots&\cdots\cr
m_{\bar {d},({\bar {d}}+1)^n-1,r}f_r&\cdots\cr
\end{block}
\end{blockarray}\
\begin{blockarray}{c}
\cr
\cr
\begin{block}{(c)}
m_{\bar D,1}\cr
m_{\bar D,2}\cr
\vdots\cr
m_{\bar D,({\bar D}+1)^n-1}\cr
\end{block}
\cr
\end{blockarray}
=
\begin{blockarray}{c}
m_{\bar D,0}=1\cr
\begin{block}{(c)}
-f_1(\mathbf0)\cr
\vdots\cr
0\cr
-f_2(\mathbf0)\cr
\vdots\cr
0\cr
\end{block}
\end{blockarray}\ ,
\end{eqnarray}
%\begin{eqnarray}\label{eq-1}
%\begin{blockarray}{cc}
%&m_1<m_2<\cdots<m_{Q_D-1}\cr
%\begin{block}{c(c)}
%m_0f_1&\cdots\cr
%\vdots&\cdots\cr
%m_0f_r&\cdots\cr
%m_1f_1&\cdots\cr
%\vdots&\cdots\cr
%m_{Q_{D-d_r}-1}f_r&\cdots\cr
%\end{block}
%\end{blockarray}\
%\begin{blockarray}{c}
%\cr
%\cr
%\begin{block}{(c)}
%m_1\cr
%m_2\cr
%\vdots\cr
%m_{Q_{D}-1}\cr
%\end{block}
%\cr
%\end{blockarray}
%=
%\begin{blockarray}{c}
%m_0\cr
%\begin{block}{(c)}
%-f_1(\mathbf0)\cr
%\vdots\cr
%-f_r(\mathbf0)\cr
%0\cr
%\vdots\cr
%0\cr
%\end{block}
%\end{blockarray}\ ,
%\end{eqnarray}
denoted as
\begin{equation}\label{eq-mls}
\MD_{\FS,D}\mv_D =\mathbf b_{\FS,D},
\end{equation}
where
$$\mv_D= (m_{\bar D,1},m_{\bar D,2},\ldots,m_{\bar D,({\bar D}+1)^n-1})^T.$$
Then  the $(i-1)(\bar d+1)^n$-th to
the $(i(\bar d+1)^n-1)$-th rows of $\MD_{\FS,D}$ are generated by $\m_{\le\bar{d}}f_i$.
The $i$-th column of $\MD_{\FS,D}$ consists of the coefficients
of $m_{\bar D,i+1}$ in $m_{\bar {d},0}f_1,\ldots,m_{\bar {d},(\bar {d}+1)^n-1}f_r$.
$\MD_{\FS,D}$ is called the {\em modified Macaulay matrix}
of the polynomial system $\FS$
and \bref{eq-mls} is called the {\em Macaulay linear system} of $\FS$.
$\MD_{\FS,D}$ is a matrix over $\C$ of dimension $(r (\bar{d}+1)^n)\times ((\bar D+1)^n-1)
= (r2^{n\delta})\times (2^{n\Delta}-1)$.

\begin{rem}\label{rem-0l}
The columns corresponding to monomial $m_{\bar D,j}$ with $\deg(m_{\bar D,j})>D$ are all $0$ columns.
\end{rem}

\begin{rem}\label{rem-02}
The zero rows are added so that the modified Macaulay matrix can be efficiently queried.
Refer to Lemma \ref{lm-nd} for details.
\end{rem}

\begin{exmp}\label{ex-1}
Let $f_1=x_1^2-x_2,\,f_2=x_2-1,f_3=x_1x_2-1,\,D=2$. Then $\bar d=1$,  $\bar D=3$ and the Macaulay linear system
is
\[
\begin{blockarray}{c(ccccccccccccccc)}
f_1\       &-1&0&0&0&0&0&0&1&0&0&0&0&0&0&0\\
0(x_2f_1)\    &0&0&0&0&0&0&0&0&0&0&0&0&0&0&0\\
0(x_1f_1)\    &0&0&0&0&0&0&0&0&0&0&0&0&0&0&0\\
0(x_1x_2f_1)\ &0&0&0&0&0&0&0&0&0&0&0&0&0&0&0\\
f_2\       &1&0&0&0&0&0&0&0&0&0&0&0&0&0&0\\
x_2f_2\    &-1&1&0&0&0&0&0&0&0&0&0&0&0&0&0\\
x_1f_2\    &0&0&0&-1&1&0&0&0&0&0&0&0&0&0&0\\
0(x_1x_2f_2)\ &0&0&0&0&0&0&0&0&0&0&0&0&0&0&0\\
f_3\       &0&0&0&0&1&0&0&0&0&0&0&0&0&0&0\\
0(x_2f_3)\    &0&0&0&0&0&0&0&0&0&0&0&0&0&0&0\\
0(x_1f_3)\    &0&0&0&0&0&0&0&0&0&0&0&0&0&0&0\\
0(x_1x_2f_3)\ &0&0&0&0&0&0&0&0&0&0&0&0&0&0&0\\
\end{blockarray}\
\begin{blockarray}{(c)}
x_2\\
x_2^2\\
x_2^3\\
x_1\\
x_1x_2\\
x_1x_2^2\\
x_1x_2^3\\
x_1^2\\
x_1^2x_2\\
x_1^2x_2^2\\
x_1^2x_2^3\\
x_1^3\\
x_1^3x_2\\
x_1^3x_2^2\\
x_1^3x_2^3
\end{blockarray}
%\begin{blockarray}{c}
=
%\end{blockarray}
\begin{blockarray}{(c)}
0\\0\\0\\0\\1\\0\\0\\0\\1\\0\\0\\0
\end{blockarray}\ .
\]
\end{exmp}

In the rest of this section, we will prove the following main result of this section,
which follows from Lemmas \ref{lem-6d}, \ref{lm-m00}, and \ref{lm-nd}.
\begin{lem}\label{lem-A}
The Macaulay linear system $\MD_{\FS,D}\mv_D =\mathbf b_{\FS,D}$
satisfied {\bf Assumptions} 1 and 2,
where the parameters are:
$M=r 2^{\delta n}$,
$N = 2^{\Delta n}-1$,
$s=T_\FS$,
$\gamma= O(n\log(D) +\log r)$,
$\rho \le  r-1$.
The parameters $\delta, \Delta,\rho, D$ are defined in  Definition \ref{def-par}.
\end{lem}
%\begin{proof}
%The values of $M$ and $N$ are from Lemma \ref{lem-6d}.
%The parameter $s$ is from Corollary \ref{cor-TF}.
%Lemmas \ref{lm-m00} and \ref{lm-nd} show that $\MD_{\FS,D}$
%satisfies {\bf Assumption} 1 with $\gamma= O(n\log(D) +\log r)$.
%%
%Without loss of generality, we may assume $f_i(0)\ne0$ for $i=1,\ldots\rho$
%and $f_i(0)=0$ for $i=\rho+1,\ldots,r$.
%Then {\bf Assumption} 2 is clearly true from the definition of $\MD_{\FS,D}$.
%\end{proof}

For nonnegative integers $B>1$ and $k<B^n$,
we denote $k_{(B)} =(k_{n-1},\ldots,k_0)$,
where $k=\sum_{i=0}^{n-1} k_iB^i$
is the $B$-base representation of $k$ and thus $0\le k_i< B$.
On the other hand, for $\mathbf k=(k_{n-1},\ldots,k_0)$ such that $0\le k_i< B$, let $\mathbf k_{(B)}=\sum_{i=0}^{n-1} k_iB^i$.
The following simple fact is crucial in the complexity analysis of
our algorithm: the $k$-th element in $\m_{\le d}$ is
\begin{equation}\label{eq-ms1}
m_{d,k}=\X^{k_{(d+1)}}=\prod_{i=0}^{n-1} x_{i+1}^{k_i}
\hbox{ with } \deg(m_{d,k}) = \sum_{i=0}^{n-1} k_i
\end{equation}
where $k_{(d+1)}=(k_{n-1},\ldots,k_0)$.
Equation \bref{eq-ms1} is true  under the assumption
made in Remark \ref{rem-sb}.

\begin{lem}\label{lem-6d}
We have $\bar d+1\le 2D$ and
$\bar D+1\le2D+1.$
As a consequence, $\MD_{\FS,D}$ is of dimension $(r (\bar{d}+1)^n)\times ((\bar D+1)^n-1) =  (r2^{n\delta})\times (2^{n\Delta}-1)
= O(r(2D)^n)\times O((2D+1)^n)$.
\end{lem}
\begin{proof}
From Definition \ref{def-par} of $\bar d$ and $\bar D$, we have
$\bar d+1\le2D-2\min_id_i+1\le 2D$ and
$\bar D+1\le2D+1.$
\end{proof}
\begin{lem}\label{lm-m00}
Let $f_i=\sum_{j=1}^{t_i}c_{ij}\X^{\alpha_{ij}}$ for $i=1,\ldots,r$, where $t_i=\#f_i$. Then $\MD_{\FS,D}$ has a natural 1-sparseness decomposition:
%regarded as a matrix function in $c_{ij}$
\begin{equation}
\MD_{\FS,D}=\sum_{i=1}^r\sum_{j=1}^{t_i}c_{ij}\MD_{ij},
\end{equation}
where each $\MD_{ij}$ is a $1$-sparse $\{0,1\}$-matrix.
In fact, only the $((i-1)(\bar d+1)^n+k,(k_{(\bar{d}+1)}+\alpha_{ij})_{(\bar D+1)})$ entries in $\MD_{ij}$ equal to $1$ for {$0\le k<(\bar{d}+1)^n$} with $\sum_{t=1}^n k_{(\bar d+1)}[t]\le D-d_i$.
\end{lem}
\begin{proof}
We can treat the coefficients $c_{ij}$ of $f_i$ as new indeterminates
and then write $\MD_{\FS,D}$ as a function in
$c_{ij}$. The coefficient of $c_{ij}$ is $M_{ij}$.
Since $c_{ij}\X^{\alpha_{ij}}$ is a term of $f_i$, only $m_{\bar{d},k}f_i$
contains terms whose coefficient are $c_{ij}$,
which corresponds to the $((i-1)(\bar d+1)^n+k)$-th row generated by  $m_{\bar{d},k}f_i$ for $0\le k<(\bar{d}+1)^n$ with $\deg(m_{\bar d,k}f_i)\le D$, or $\sum_{t=1}^n k_{(\bar d+1)}[t]\le D-d_i$ equally. Since $m_{\bar{d},k}\cdot\X^{\alpha_{ij}}=\X^{k_{(\bar{d}+1)}}
\cdot\X^{\alpha_{ij}}=\X^{k_{(\bar{d}+1)}+\alpha_{ij}}=m_{\bar D,({k_{(\bar{d}+1)}+\alpha_{ij}})_{(\bar D+1)}}$,  only the $({k_{(\bar{d}+1)}+\alpha_{ij}})_{(\bar D+1)}$-th entry in the $((i-1)(\bar d+1)^n+k)$-th row of $\MD_{ij}$ is nonzero.
%We complete the proof.
%
The $1$-sparseness of $\MD_{ij}$ comes from the fact that
the rows and columns of the nonzero entries are distinct.
\end{proof}

As a direct consequence, we have
\begin{cor}\label{cor-TF}
$\MD_{\FS,D}$ is $T_\FS$-sparse, where $T_\FS=\sum_{i=1}^rt_i$ is called the
{\em total sparseness} of $\FS$.
\end{cor}

\begin{lem}\label{lm-nd}
The complexity of a query for $\MD_{ij}$ is $O(n\log(D) +\log r)$,
where $\MD_{ij}$ is introduced in Lemma \ref{lm-m00}.
%Moreover, we can construct an oracle to query $\MD_{ij}$ in time $O(n\log(\bar D) +\log r)$.
%
\end{lem}
\begin{proof}
Given a row index $i_0$, we want to know the nonzero entry in the $i_0$-th row. By Lemma \ref{lm-m00}, the $i_0$-th row has a nonzero entry if and only if {$(i-1)(\bar d+1)^n\le i_0< i(\bar d+1)^n$} and $\sum_{t=1}^n (i_0-(i-1)(\bar d+1))_{(\bar d+1)}[t]\le D-d_i$.
Compute the quotient $l$ and remainder $k$ such that $i_0=(\bar d+1)^nl+k$.
By Lemma \ref{lm-m00}, the $i_0$-th row has a nonzero entry if and only if $l=i-1$ and $\sum_{t=1}^n k_{(\bar d+1)}[t]\le D-d_i$, meanwhile the nonzero entry is at the $(k_{(\bar{d}+1)}+\alpha_{ij})_{(\bar D+1)}$-th column.

We now analyse the complexity of the above step.
Without loss of generality, we assume that all the
numbers are represented in binary form,
which is crucial to the complexity.
Since $\bar d+1=2^\delta$ is a power of $2$, we can compute $(\bar{d}+1)^n=2^{\delta n}$ easily in time $O(\log(n)\log(\log(\bar{d})))$. We can check $\sum_{t=1}^n k_{(\bar d+1)}[t]\le D-d_i$ in time $\log(n\log D)$. Since all numbers are binary, the last $\delta n$ bits of $i_0$ are exactly the remainder $k$, and other bits are the quotient $l$. As a result, we can compute both $l$ and $k$ in time $O(\log(n)\log(\log(\bar{d})))$.
Since the number of bits for $i_0$ is $O(\log i_0)=O(\log(r(\bar d+1)^n)=O(\log r+n\log\bar d)$, the complexity is bounded by $O(\log r+n\log\bar d)$.
Since both $\bar{d}+1=2^\delta$ and $\bar D+1=2^\Delta$ are powers of 2 and $k$ is in binary form, we can insert $(\Delta-\delta)$ zeros before each $\delta$ bits of $k$ starting from lower digits  to obtain $(k_{(\bar{d}+1)})_{(\bar D+1)}$ in time $O(n\log(\bar D))$. Totally, we can compute $({k_{(\bar{d}+1)}+\alpha_{ij}})_{(\bar D+1)}$ in time $O(\log(n)\log(\log(\bar D))+n\log(\bar D))=O(n\log(\bar D))$.
So the total complexity is $
O(n\log(\bar D)+(\log r+n\log\bar d))=O(n\log(D) +\log r)$ by Lemma \ref{lem-6d}.

On the other hand, given a column index $j_0$, we want to know the nonzero entries in the $j_0$-th column. Compute $\mathbf k=(j_0)_{(\bar D+1)}-\alpha_{ij}$ first.
If $\mathbf k$ is a nonnegative vector, we can compute $\mathbf k_{(\bar{d}+1)}$. By Lemma \ref{lm-m00}, only if $\sum_{t=1}^n\mathbf k[t]\le D-d_i$, the $(\mathbf k_{(\bar{d}+1)}+(j-1)(\bar{d}+1)^n)$-th entry is the unique nonzero entry in the $j_0$-th column. Else the $j_0$-th column is a 0 column. Similarly, the complexity is $O(n\log(D) +\log r)$.
\end{proof}

\subsection{Complete solving degree for a polynomial system}
\label{sec-ps2}
In this section, we show how to determine a $D$ such that
the monomials of a polynomial system $\FS$ can be solved from
the Macaulay linear system $\MD_{\FS,D}\mv_D =\mathbf b_{\FS,D}$
by introducing the concept of complete solving degree.
We first define the concept of solving degree \cite{laz83,soldeg-1},
which is an important concept for polynomial system solving.

\begin{defn}\label{def-sdeg}
Let $\FS=\{f_1,\ldots,f_r\}\subset\C[\X]$
and $(\FS)$ the ideal generated by $\FS$.
Let $\mathbb G$ be the reduced Gr\"obner basis of the ideal $(\FS)$ under the degree reverse lexicographic (DRL) monomial ordering.
$D$ is called the {\em solving degree} of $\FS$, if $D$ is the minimal integer such that for any  $g\in\mathbb G$, $g=\sum_{i=1}^rh_if_i$ for some $h_1,\ldots,h_r\in\C[\X]$ satisfying $\deg(h_if_i)\le D$.
Denote the solving degree of $\FS$ by $\sdeg(\FS)$.
\end{defn}
In terms of the F4 algorithm \cite{f4} or the XL algorithm \cite{xl},
$D$ is the solving degree of $\FS$, if the Gr\"obner basis of the ideal
$(\FS)$ can be obtained from $\MD_{\FS,D}\mv_D =\mathbf b_{\FS,D}$ by using Gaussian elimination over $\C$.
%
%By definition, it is easy to check that $\csdeg(\FS)\ge\sdeg(\FS)$.

For a polynomial $f\in\C[\X]$, denote $f^h$ to be the homogenization of $f$ in $\C[x_0,\X]$.
We have the following upper bound for the solving degree.
\begin{thm}[\cite{soldeg-1} Corollary 3.26]\label{cor-sdb}
Let $I=(f_1,\ldots,f_r)\subset\C[\X]$ be an ideal generated by $f_1,\ldots,f_r$ of degrees $d_1,\ldots,d_r$, such that $d_1\ge d_2\ge\cdots\ge d_r$.
Choose the DRL monomial ordering. If $I^h=(f_1^h,\ldots,f_r^h)$
is zero-dimensional, then  $\sdeg(\FS)\le d_1+\cdots+d_{n+1}-n+1$ with $d_{n+1}=1$ if $r=n$.
\end{thm}

The following example shows that
the solving degree is not large enough for monomial solving with the Macaulay linear system.
\begin{exmp}\label{ex-csdeg}
Let $\FS=\{x_1x_2-1,x_1^2-x_1,x_2^2-x_2,x_3^2-x_3\}\subset\C[x_1,x_2,x_3]$.
The reduced Gr\"obner basis  for $(\FS)$ is
$\mathbb G=\{x_1-1,x_2-1,x_3^2-x_3\}$.
Since $x_1-1=x_2(x_1^2-x_1)+(1-x_1)(x_1x_2-1)$ and $x_2-1=x_1(x_2^2-x_2)+(1-x_2)(x_1x_2-1)$, we have $\sdeg(\FS)=3$.
The residue monomials of $(\FS)$ with respect to  $\mathbb G$ is $\{1, x_3\}$.
Let $D=3$, we want to solve the monomials
of degrees $\le D$ from the Macaulay linear system $\MD_{\FS,D}\mv_D =\mathbf b_{\FS,D}$,
or more precisely, we want to written all monomials of degrees ${\le 3}$ as expressions of $x_3$.
After doing Gaussian elimination to the Macaulay system, the nontrivial polynomials
are:
$\GB_3=\{x_1 - 1, x_2 - 1, x_1^2 - 1, x_1 x_2 - 1, x_2^2 - 1,
 x_3^2 - x_3, x_1^3 - 1, x_1^2 x_2 - 1,
 x_1^2 x_3 - x_1 x_3, x_1 x_2^2 - 1, x_1 x_2 x_3 - x_3,
 x_1 x_3^2 - x_1 x_3, x_2^3 - 1, x_2^2 x_3 - x_2 x_3,
 x_2 x_3^2 - x_2 x_3, x_3^3 - x_3\}$.
We can see that each monomial of degree $\le 3$ is solved in terms of
$x_3$, $x_1x_3$, and $x_2x_3$,
which means that the Macaulay linear system does not give the ``correct"
solution to the monomials and the monomials $x_1x_3$ and $x_2x_3$ are not solved in terms of $x_3$.
%
%We have $\csdeg(\FS)=4\ne3$ because neither $x_1x_3$ nor $x_2x_3$ is a leading term for any linear combinations of $\m_{\le1}\FS$.
\end{exmp}

Motivated by the above example, we introduce the concept of complete solving degree.
\begin{defn}\label{def-csdeg}
Let $\FS=\{f_1,\ldots,f_r\}\subset\C[\X]$
and $(\FS)$ the ideal generated by $\FS$.
Let $\mathbb G$ be the reduced Gr\"obner basis of the ideal $(\FS)$ under the DRL monomial ordering.
$D$ is called the {\em complete solving degree} of $\FS$, if $D$ is the minimal integer such that for any polynomial $g\in\mathbb G$ and any monomial $m\in\C[\X]$ satisfying $\deg(mg)\le D$, $mg=\sum_{i=1}^rh_if_i$ for some $h_1,\ldots,h_r\in\C[\X]$ with each $\deg(h_if_i)\le D$.
Denote the complete solving degree of $\FS$ by $\csdeg(\FS)$.
\end{defn}

\begin{exmp}\label{ex-csdeg1}
For the $\FS$ given in Example \ref{ex-csdeg}, we have $\csdeg(\FS)=4$.
% by Lemma \ref{lm-csdeg}, we have $\csdeg(\FS)\le8$. Actually,  $\csdeg(\FS)=4$.
For $D=4$, after doing Gaussian elimination to the Macaulay linear system, we
obtain $\GB_4=\GB_3\cup\{
x_1 x_3 - x_3,
x_2 x_3 - x_3,
x_1^4 - 1, x_1^3 x_2 - 1, x_1^3 x_3 - x_3,
 x_1^2 x_2^2 - 1, x_1^2 x_2 x_3 - x_3, x_1^2 x_3^2 - x_3,
 x_1 x_2^3 - 1, x_1 x_2^2 x_3 - x_3, x_1 x_2 x_3^2 - x_3,
  x_1 x_3^3 - x_3, x_2^4 - 1, x_2^3 x_3 - x_3,
 x_2^2 x_3^2 - x_3, x_2 x_3^3 - x_3, x_3^4 - x_3\}$.
Note that all monomials with degree $\le 4$
are written as expressions of the residue monomials $1, x_3$.
In other words, the monomials are solved with the Macaulay system
for $D=4$.
\end{exmp}

In this paper, we will consider polynomial systems of the following form
\begin{eqnarray}
\FS&=&\FS_1\cup\FS_2\subset \C[\X], \hbox{ where }
\FS_1=\{g_1,\ldots,g_{r}\}\ne\emptyset,
\FS_2=\{f_1,\ldots,f_n\}\label{eq-FSB}\\
&&\hbox{ such that } \forall i,j\,
(\lm(f_i)=x_i^{d_i}, d_i\ge 1, \hbox{ and } \deg_{x_i}(g_j) < d_i)\nonumber
\end{eqnarray}
where $\lm(f_i)$ is the largest monomial of $f_i$ under the DRL monomial ordering, also called the {\em leading monomial} of $f_i$. We will give upper bounds for the solving degree
and complete solving degree for a polynomial system of form \bref{eq-FSB}
in the following lemmas. Note that the proofs of these results do not depend on the results in \cite{laz83,soldeg-1}.

\begin{lem}\label{lm-sdnew}
For $\FS$ given in \bref{eq-FSB}, we have $(\FS)$ is zero-dimensional and $\sdeg(\FS)\le d-n+\sum_{i=1}^n d_i$,
where $d= \max_i\deg(g_i)$.
\end{lem}
\begin{proof}
By the Hilbert function,
it is easy to see that $(\FS_2)$ is zero-dimensional and hence $(\FS)$ is zero-dimensional.
Let $I=(\FS)$ and $\GB$ the reduced Gr\"obner basis for $(\FS)$ under the DRL monomial ordering. Denote $\lm(I)$ to be the set of the leading monomials of the polynomials in $I$.
Since $x_i^{d_i}\in\lm(I)$, there exists a unique polynomial $f_i'$ in $\GB$ such that $\lm(f_i')=x_i^{d_i'}$ with $d_i'\le d_i$.
For any $h\in\GB$, let $h=\sum_{k=1}^{n}a_kf_k+\sum_{l=1}^{r}b_lg_l$.
We claim that $\deg(h)\le\max\{d_1,\ldots,d_n,\sum_{i=1}^n d_i-n\}$.
If $h$ is reduced by each $f_k'$, we have $\deg_{x_k}(h)<\deg_{x_k}(f_k')=d_k'\le d_k$ and $\deg(h)\le\sum_{i=1}^n d_i-n$.
Otherwise, $h$ is exactly $f_k'$ for some $k$, and then $\deg(h)=d_k'\le d_k$.
Thus we have $\deg(h)\le\max\{d_k,\sum_{i=1}^n d_i-n\}$ and the claim is proved.

If $b_j$ is not reduced by some $f_i$, we have $b_j=pf_i+q$ where $q$ is reduced by $f_i$ and $h=(\sum_{k=1,k\ne i}^{n}a_kf_k+(pg_j+a_i)f_i)+(\sum_{l=1,l\ne j}^{r-n}b_lg_l+qg_j)$. Still write the new expression as $h=\sum_{k=1}^{n}a_kf_k+\sum_{l=1}^{r}b_lg_l$.
As a consequence, we can assume that each $b_j$ is reduced by $\FS_2$ and hence $\deg(b_l)\le\sum_{i=1}^n d_i-n$.

Let $\tilde h=h-\sum_{l=1}^{r}b_lg_l=\sum_{k=1}^{n}a_kf_k\in(\FS_2).$
From \bref{eq-FSB}, it is easy to see that  $\FS_2$ is a Gr\"obner basis for $(\FS_2)$.
Then, there exist $\tilde a_1,\ldots,\tilde a_n$ such that $\tilde h=\sum_{k=1}^{n}\tilde a_kf_k$ with $\deg(\tilde a_kf_k)\le\deg(\tilde h)\le\max\{h,b_lg_l\}$.
Thus we have $h=\sum_{k=1}^{n}\tilde a_kf_k+\sum_{l=1}^{r-n}b_lg_l$.
By the definition of solving degree, we have $\sdeg(\FS)\le\max\{\deg(\tilde a_kf_k),\deg(b_lg_l)\}\le\max\{\deg(\tilde h),\deg(b_lg_l)\}\le\max\{\deg(h),\deg(b_lg_l)\}\le\max\{d_1,\ldots,d_n,\sum_{i=1}^n d_i-n,\sum_{i=1}^n d_i-n+d\}=\sum_{i=1}^n d_i-n+d$, where the last equation comes from the assumption $d_i\ge1$.
\end{proof}

\begin{lem}\label{lm-csd}
For $\FS$ given in \bref{eq-FSB}, we have $\csdeg(\FS)\le d-2n+2\sum_{i=1}^n d_i$,
where $d= \max_i\deg(g_i)$.
\end{lem}
\begin{proof}
Let $I =(\FS)$ and $D=d-2n+2\sum_{i=1}^n d_i$.
Fix a monomial $m\in\lm(I)$ with $\deg(m)\le D$.
We first claim that  there exists an $h= \sum_i p_i f_i +\sum_jq_jg_j$ for some  $p_i,q_j\in\C[\X]$ such that
$\lm(h)=m$, $\deg(p_i f_i)\le D$, and $\deg(q_jg_j) \le D$.
If $m$ is a multiple of some $\lm(f_i)$,
we have $m=\lm(\frac{m}{\lm(f_i)}f_i)$ with $\deg(\frac{m}{\lm(f_i)}f_i)\le D$
and the claim is proved.
Otherwise, we have $\deg_{x_i}(m)<d_i$ for each $i$ and thus $\deg(m)\le\sum_{i=1}^nd_i-n$.
There exists a $g$ in the reduced Gr\"obner basis $\GB$ of $(\FS)$, such that $\lm(g)|m$. Let $g=\sum_ka_kf_k+\sum_lb_lg_l$ with $\deg(a_kf_k),\deg(b_lg_l)\le\sdeg(\FS)$,
and we have $m=\lm(\frac{m}{\lm(g)}g)=\lm(\sum_k\frac{m}{\lm(g)}a_kf_k
+\sum_l\frac{m}{\lm(g)}b_lg_l)$ with $\deg(\frac{m}{\lm(g)}a_kf_k),\deg(\frac{m}{\lm(g)}b_lg_l)
\le\deg(m)+\sdeg(\FS)\le\sum_{i=1}^nd_i-n+\sdeg(\FS)\le D$, where the last inequality comes from Lemma \ref{lm-sdnew} and the assumption $\deg(m)\le\sum_{i=1}^nd_i-n$.
The claim is proved.

We now prove the lemma. Fix a polynomial $g\in\mathbb G$ and a monomial $m\in\C[\X]$ satisfying $\deg(mg)\le D$. Since $\lm(mg)\in\lm(I)$ and $\deg(mg)\le D$,
by the claim just proved, there exists an $h= \sum_i p_i f_i +\sum_jq_jg_j$ for some  $p_i,q_j\in\C[\X]$ such that
$\lm(mg) = \lm(h)$, $\deg(p_i f_i)\le D$, and $\deg(q_jg_j) \le D$.
Let $\widetilde{g} = mg - \frac{\lc(mg)}{\lc(h)}h$, where $\lc(p)$ is the coefficient of
$\lm(p)$ in $p$ for any $p\in\C[\X]$.
Then $\widetilde{g}\in I$ and $\deg(\lm(\widetilde{g}))< D$.
We thus can repeat the above procedure for $\widetilde{g}$ (instead of $mg$).
The process will end and we obtain
$\widehat{g} = mg - \sum_{i=1}^n \widehat{p}_i f_i-\sum_{j=1}^r \widehat{q}_jg_j$,
where $\lm(\widehat{g})\not\in\lm(I)$, $\deg(\widehat{p}_i f_i)\le D$,
and $\deg(\widehat{q}_jg_j) \le D$.
Since $\lm(\widehat{g})\not\in\lm(I)$ and $\widehat{g}\in I$, we have
$\widehat{g}=0$ and hence $mg = \sum_{i=1}^n \widehat{p}_i f_i+\sum_{j=1}^r \widehat{q}_jg_j$. The lemma is proved.
%
%There exist some $m_jf_i$ or $m_jg_i$ with degrees $\le D$ such that $\lm(mg)=\lm(m_jf_i)$ (or $\lm(m_jg_i)$) by the claim aforementioned.
%%
%Repeat for $mg-\frac{\lc(mg)}{\lc(m_jf_i)}m_jf_i$, and we have $\tilde g=mg-\sum c_*m_jf_i-\sum c_*m_jg_i$ satisfying either $\lm(\tilde g)\notin\lm(I)$ or $\tilde g=0$. Since $\tilde g\in I$ implies $\lm(\tilde g)\in\lm(I)$, $\tilde g$ must be $0$.
%%
%Thus $mg=\sum c_*m_jf_i+\sum c_*m_jg_i$ with each $\deg(m_jf_i)$,$\deg(m_jg_i)\le D$. By the definition for complete solving degree, $\csdeg(\FS)\le D=d-2n+2\sum_{i=1}^n d_i$.
\end{proof}

\subsection{Solution of the Macaulay linear system }

The following lemma gives the solutions to the
Macaulay linear system \bref{eq-1} by using the complete solving degree.

\begin{lem}\label{sol-sh}
Let $I=(\FS)$ be a radical zero-dimensional ideal in $\C[\X]$ generated by $\FS$ and $\V_\C(I)=\{\mathbf{a}_1,\ldots,\mathbf{a}_w\}$.
For $D\ge\csdeg(\FS)$, any solution $\mv_D$ of the Macaulay linear system $\MD_{\FS,D}\mv_D=\mathbf b_{\FS,D}$ is of the form
$$\widehat{\mv}_D= \sum\limits_{i=1}^w \eta_i{\mv_D}(\mathbf a_i)+\sum\limits_{\sum_ik_{(\bar D+1)}[i]>D}\mu_{k}\mathbf e_{k-1},$$
where $\eta_i$ are complex numbers such that $\sum_{i=1}^w \eta_i = 1$, $\mu_{k}$ are arbitrary complex numbers, and $\mathbf e_k$ is the $k$-th unit vector in $\C^{(\bar D+1)^n-1}$.
\end{lem}
\begin{proof}
From \bref{eq-ms1}, we have $\deg(m_{\bar D,k}) = \sum_ik_{(\bar D+1)}[i]$.
By {Remark \ref{rem-0l}}, if $\deg(m_{\bar D,k})=\sum_ik_{(\bar D+1)}[i]$ $>D$, then the $(k-1)$-th row in $\MD_{\FS,D}$ corresponded to $m_{\bar D,k}$  is a zero column,
so $m_{\bar{D},k}$ can take arbitrary value in the solution of the Macaulay linear system
and hence $\mu_{k}\mathbf e_{k-1}$ is a solution of the Macaulay linear system.
Delete the $(k-1)$-th column in $\MD_{\FS,D}$ and the $(k-1)$-th row in $\mv_D$ and $\mathbf b_{\FS,D}$ to obtain a new system $\widetilde{\MD}_{\FS,D}\widetilde{\mv}_{\FS,D}=\widetilde{\mathbf b}_{\FS,D}$.

When we do Gaussian elimination on the linear system $\widetilde{\MD}_{\FS,D}\widetilde{\mv}_{\FS,D}=\widetilde{\mathbf b}_{\FS,D}$, we get a new linear system $\overline{\MD}_{\FS,D}\widetilde{\mv}_{\FS,D}= \overline{\mathbf b}_{\FS,D}$ .
Since $D\ge\csdeg(\FS)$, the largest monomial in each row of $\overline{\MD}_{\FS,D}$ is in $\lm(I)$
and each monomial in $\lm(I)$ occurs as one of the leading monomials for some row.
Thus, the dimension of the solution space of $\overline{\MD}_{\FS,D}\widetilde{\mv}_{\FS,D}=\overline{\mathbf b}_{\FS,D}$ is at most
$\#(\m_{\le D}\setminus\lm(I))-1\le\#(\m\setminus\lm(I))-1=\#V(I)-1=w-1$, where the first equality is true
because $I$ is radical.
Since each $\widetilde{\mv}_{D}(\mathbf{a}_i)$ is a solution of $\overline{\MD}_{\FS,D}\widetilde{\mv}_{D}=\overline{\mathbf b}_{\FS,D}$, $\{\sum\eta_i\widetilde{\mv}_{D}(\mathbf a_i)|\sum\eta_i=1\}$ is a subspace of the solution space of $\overline{\MD}_{\FS,D}\widetilde{\mv}_{D}=\overline{\mathbf b}_{\FS,D}$,
that is, the solution space is of dimension at least $w-1$.
Then, $\{\sum\eta_i\widetilde{\mv}_{D}(\mathbf a_i)|\sum\eta_i=1\}$ is exactly the solution space of $\overline{\MD}_{\FS,D}\widetilde{\mv}_{D}=\overline{\mathbf b}_{\FS,D}$, also that of $\widetilde{\MD}_{\FS,D}\widetilde{\mv}_{D}=\widetilde{\mathbf b}_{\FS,D}$.
The lemma is proved.
\end{proof}
\begin{cor}\label{cr-unis}
In Lemma \ref{sol-sh}, if $\FS$ has a unique solution $\mathbf a$, then
we have
$$\widehat{\mv}_D=\mv_D(\mathbf a)+\sum\limits_{\sum_ik_{(\bar D+1)}[i]>D}\mu_{k}\mathbf e_{k-1}.$$
\end{cor}
\begin{exmp}\label{ex-101}
The equation  system in {Example \ref{ex-1}} has a unique solution $\mathbf a=(1,1)$, with $\mv_D(\mathbf a_1)=(1,1,1,1,1,1,1,1,1,1,1,1,1,1,1)$.
By Lemma \ref{sol-sh},
the solution of the Macaulay linear system is $(1,1,1+\mu_3,1,1,1+\mu_6,
1+\mu_7,1,1+\mu_9,1+\mu_{10},1+\mu_{11},
1+\mu_{12},1+\mu_{13},
1+\mu_{14},1+\mu_{15})$,
where each $\mu_i$ is an arbitrary complex number.
The solutions of the form $1+\mu_{k}$ correspond to the zero
columns of $\MD_{\FS,D}$.
\end{exmp}

\subsection{A quantum algorithm for monomial-solving of polynomial systems}
In this section, we show that applying the HHL algorithm
to the Macaulay linear system of $\FS$, we obtain a pseudo solution
of $\FS$.

Let $\FS\subset\C[\X]$ with $\V_\C(\FS)=\{\mathbf a_1, \ldots, \mathbf a_w\}$
and ${{\mv}}$ a vector of monomials in $\X$.
A normalized linear combination of $\sum_i c_i {\mv}(\mathbf a_i)$
such that $\sum_i c_i=1$ is also a solution to the
Macualay linear system,
which is called a {\em pseudo solution} of $\FS$.

Note that a {\em pseudo solution} of $\FS$ is a solution of the Macualay linear system.
From Lemma \ref{sol-sh}, the converse is not rue: a solution of the Macaulay linear system contains some arbitrary constants $\mu_k$ and is not a pseudo solution of $\FS$.
Fortunately, applying the HHL algorithm to the Macaulay  linear system,
the solution of the Macualay linear system is a pseudo solution of $\FS$.

We introduce the notation $\widetilde{\mv}\in\m^{(\bar D+1)^n-1}$:
for $k=1,\ldots,(\bar D+1)^n-1$,
\begin{equation}\label{eq-rm}
\widetilde{\mv}_D[k] =\begin{cases}
m_{\bar D,k},& \text{if } \deg(m_{\bar D,k})= \sum_ik_{(\bar D+1)}[i]\le D;\\
0,& \text{if } \deg(m_{\bar D,k})= \sum_ik_{(\bar D+1)}[i]>D.
\end{cases}
\end{equation}

Using the modified HHL algorithm (Theorem \ref{thm-hhl}) to the Macaulay linear system, we have
\begin{thm}\label{th-qsh}
Let $I=(\FS)$ be a radical zero-dimensional ideal in $\C[\X]$ generated by $\FS$, $\V_\C(I)=\{\mathbf{a}_1,\ldots,\mathbf{a}_w\}$, $\epsilon\in(0,1)$,
and $D\ge\csdeg(\FS)$.
%
%Let $D=d-2n+2\sum_{i=1}^n d_i$ and $\epsilon\in(0,1)$.
Using the modified HHL algorithm to the Macaulay linear system $\MD_{\FS,D}\mv_D = \mathbf b_{\FS,D}$ defined in \bref{eq-mls}, the answer is an $\epsilon$-approximation to the following pseudo solution of $\FS$
$$\widehat{\mv}_D= \sum\limits_{i=1}^w \eta_i\widetilde{\mv}_D(\mathbf a_i),$$
where the monomial vector  $\widetilde{\mv}_D$ is defined in  \bref{eq-rm},
$\eta_i$ are complex numbers satisfying $\sum_{i=1}^w \eta_i = 1$, and $\|\sum_{i=1}^w\eta_i\widetilde{\mv}_D(\mathbf a_i)\|$
is minimal.
The runtime complexity is $\widetilde{O}(\log(D)nT_\FS\kappa^2/\epsilon)$,
where $T_\FS=\sum_{i=1}^r \#f_i$ and $\kappa$ is the condition number of $\MD_{\FS,D}$.
\end{thm}
\begin{proof}
By {Lemma \ref{sol-sh}}, $$\widehat{\mv}_D=
\sum\limits_{\sum\eta_i=1}\eta_i{\mv_D}(\mathbf a_i)+\sum\limits_{\sum_ik_{(\bar D+1)}[i]>D}\mu_{k}\mathbf e_{k-1}
=\sum\limits_{\sum\eta_i=1}\eta_i\widetilde{\mv}_D(\mathbf a_i)+\sum\limits_{\sum_ik_{(\bar D+1)}[i]>D}\widetilde\mu_{k}\mathbf e_{k-1}.$$
By Lemma \ref{lem-HHL}, $\|\widehat{\mv}_D\|$ is minimal.
Since $\langle\mathbf e_{k-1}|\widetilde{\mv}_D(\mathbf a_i)\rangle=0$,
in order for $\|\widehat{\mv}_D\|$ to be minimized, each $\widetilde\mu_{k}=0$.
We proved that the solution is indeed a pseudo solution.

We now analyse the complexity.
We may change the order of $f_i$ such that, the first $\rho$ polynomials
$f_i$ have nonzero constant terms. This step costs $O(r)$.
By Lemma \ref{lem-A}, $\MD_{\FS,D}\mv_D = \mathbf b_{\FS,D}$ satisfies {\bf
Assumptions} 1 and 2 for $s=T_\FS=\sum_{i=1}^rt_i$, $\gamma=n\log D+\log r$,
and $v= n\delta$.
%
%Since we assume that the constant terms of $f_i$ are either $0$ or $-1$, by Lemma \ref{lem-6d}, $\b_{\FS,D}$ satisfies {\bf Assumption} 2 for $v= n\delta$.
%
By Lemma \ref{lem-A},  $\MD_{\FS,D}$ is of dimension $O(r(2D)^n)\times O((2D+1)^n)$.
By {Theorem \ref{thm-hhl}} and {Lemma \ref{lm-nd}}, the complexity of the HHL algorithm is $
\widetilde{O}((\log(M+N)+\gamma)s\kappa^2/\epsilon)=
\widetilde O((\log(r(2D)^n+ (2D+1)^n)+(n\log( D)+\log r))T_\FS\kappa^2/\epsilon)
 =\widetilde{O}((\log(r)+n\log( D))T_\FS\kappa^2/\epsilon)
 =\widetilde{O}(\log(r)T_\FS\kappa^2/\epsilon+n\log( D)T_\FS\kappa^2/\epsilon)
 =\widetilde{O}(T_\FS\kappa^2/\epsilon+n\log( D)T_\FS\kappa^2$ $/\epsilon) =\widetilde O(\log(D)nT_\FS\kappa^2/\epsilon)$, since $r \le T_\FS$.
We prove the theorem.
\end{proof}
%
%In Theorem \ref{th-qsh}, the solution is a linear combination of
%the monomials of the solutions of $\FS=0$, which is called the {\color{red}\em pseudo-monomial-solution} for $\FS=0$.
%

%\begin{cor}\label{cor-cmp}
%Set $D=(n+1)(d-1)+2$ in Theorem \ref{th-qsh}, the complexity is $\widetilde O(\log(d)nT_\FS\kappa^2/\epsilon)$, where $d=\max_i \deg(f_i)$.
%\end{cor}
%\begin{proof}
%$\widetilde O(\log(n+nd)\min\{n,nd\}(\sum_{i=1}^rt_i)\kappa^2/\epsilon)=
%\widetilde O(\log(nd)n(\sum_{i=1}^rt_i)\kappa^2/\epsilon)=\widetilde %O(\log(d)n(\sum_{i=1}^rt_i)\kappa^2/\epsilon)$
%\end{proof}
%

By Remark \ref{rem-HHL},
the exact complexity for Theorem \ref{th-qsh} is
\begin{cor}\label{cor-E}
The exact complexity to compute $|\widehat{\mv}_D\rangle$
is $c\log(M+N)T_\FS\kappa^2/\epsilon$, where $N=r(\bar d+1)^n$,
$M=(\bar D+1)^n -1$, and $c$ is the complexity constant of the HHL algorithm.
\end{cor}

\begin{cor}
If $\FS$ satisfies the conditions in Theorem \ref{th-qsh} and has a unique solution $\mathbf a$, then the solution state is
$|\mv\rangle=|\widetilde{\mv}(\mathbf a)\rangle.$
The complexity is $\widetilde{O}(\log(D)nT_\FS\kappa^2/\epsilon)$.
\end{cor}

\begin{exmp}\label{ex-11}
If using the modified HHL algorithm to solve the linear system in {Example \ref{ex-1}}, by {Theorem \ref{th-qsh}}, the solution is $(1,1,0,1,1,0,0,1,0,0,0,0,0,0,0)$.
In order to find the unique solution $x_1=1,x_2=1$, we need to know how to
project $(x_2,x_2^2,x_2^3,x_1,x_1x_2,x_1x_2^2,x_1x_2^3,
x_1^2,$ $x_1^2x_2,$ $x_1^2x_2^2,x_1^2x_2^3,x_1^3,x_1^3x_2,x_1^3x_2^2,
x_1^3x_2^3)$ to $(x_1,x_2)$ efficiently.
\end{exmp}

Motivated by the above example, we propose the following problem.
\begin{prob}\label{pro-11}
Let $|u\rangle$ be an $N$-dimensional quantum state and $n \ll N$.
How can we measure $n$ selected coordinates of $|u\rangle$ {efficiently}.
\end{prob}

%The solution to the above linear system is $(x_1,x_2,x_1^2,x_1x_2,x_2^2)=(2,4,4,8,16+\alpha)$, where $\alpha$ is an arbitrary complex number.
%%
%Note that $\ND_{\FS,D} = \{x_2^2\}$.

\section{Find Boolean solutions for polynomial systems in $\C[\X]$}
\label{sec-bs}
In this section, we will give a quantum algorithm to compute the Boolean solutions of a polynomial system $\FS$ over $\C$. The key idea is that by measuring the
pseudo solutions of $\FS$, we may obtain a Boolean solution of
$\FS$ with high probability.

\subsection{A quantum algorithm to find Boolean solutions}

A solution $\mathbf a$ for $\FS\subset\C[\X]$ is called {\em Boolean}, if each coordinate of $\mathbf a$ is $0$ or $1$.
For $\FS\subset\C[\X]$, the set of Boolean solutions of $\FS$ are denoted as $\V_{B}(\FS)$.
%In other words, $\V_{\F_2}(\FS)$ is the set of {\em Boolean solutions} of $\FS$.
%
We first prove a lemma.
\begin{lem}\label{lm-la1}
For $\FS\subset\C[\X]$, we have  $\V_{B}(\FS)=\V_\C(\FS,\mathbb H_\X)=\V_\C(\FS_B,\mathbb H_\X)$ and $I=(\FS_B,\mathbb H_\X)$ is radical, where $\mathbb H_{\X}=\{x_1^2-x_1,\ldots,x_n^2-x_n\}$
and
$\FS_B$ is   obtained from $\FS$ by replacing $x_i^m$ in $\FS$ with $x_i$ for all $i$ and $m\in\N$.
Furthermore, the complete solving degree $\csdeg(\FS_B\cup \H_\X) \le 3n$.
\end{lem}
\begin{proof}
The first assertion is easy.
Since $(\FS_B)+(x_1-a_1,\ldots,x_n-a_n)$ is a maximal ideal in the ring $\C[\X]$,
%CLO Using Algebraic Geometry page 41 proposition 2.7
$$(\FS_B,\H_\X)=\bigcap\limits_{a_1,\ldots,a_n\in\{0,1\}}((\FS_B)+(x_1-a_1,\ldots,x_n-a_n)),$$
is an {intersection} of maximal ideals and $I$ is a radical ideal.
%
%$\H_\X$ satisfies the condition for $\FS_2$. With loss of generality, we can assume each monomial in $\FS$ is square-free, that is, reduced by $\H_\X$.
%
$\FS_B\cup\H_\X$ is clearly of form \bref{eq-FSB}.
By Lemma \ref{lm-csd},
$\csdeg(\FS_B\cup\H_\X)\le d-2n+2\sum_{i=1}^n d_i \le 3n$ since $d\le n$ and $d_i=2$ for all $i$.
\end{proof}

By Lemma \ref{lm-la1}, we can use {Theorem \ref{th-qsh}} to compute $|\widehat{\mv}_D\rangle$ where a bound $D=3n$ for the complete solving degree is given%in Lemma  \ref{lm-csd}
. Denote $\mathbf 0=(0,\ldots,0)^T$ and $\mathbf 1=(1,\ldots,1)^T$. Our quantum algorithm to compute Boolean solutions is given below.

\begin{alg}\label{alg-b}
\end{alg}

{\noindent\bf Input:}
$\FS=\{f_1,\ldots,f_r\}\subset \mathcal\C[\X]$ with $T_\FS=\sum_{i=1}^r\#f_i$ and $d_i =\deg(f_i)$. Also $\epsilon\in(0,1)$.

{\noindent\bf Output:}
 A Boolean solution $\mathbf a\in\V_{\C}(\FS,\H_\X)$  or $\emptyset$ meaning that $\V_{\C}(\FS,\H_\X)=\emptyset$, with success probability at least $1-\epsilon$.

\begin{description}

\item[Step 1:]
If $\FS(\mathbf 0)=\mathbf 0$, return $\mathbf 0$.
If $\FS(\mathbf 1)=\mathbf 0$, return $\mathbf 1$.
Set $l=1$.

%Let $\FS=PolynomialRemainder(\FS,\H_\X)$.

\item[Step 2:]
Let $\FS_1=\FS_B$  and $\Y=\X$.

\item[Step 3:]
Let $\FS_2=\FS_1\cup\H_\Y$ and $D= 3\#\Y$
(From Lemma \ref{lm-la1}, $D\ge \csdeg(\FS_2)$).
%{Theorem \ref{th-laz}}.

\item[Step 4:]
Use the modified HHL algorithm (Theorem \ref{thm-hhl}) to the Macaulay linear system
 $\MD_{\FS_2,D}$ $\mv_D$ $ =\mathbf b_{\FS_2,D}$
 to obtain a state $|\widehat{\mv}_D\rangle$
 with the error bound $\sqrt{\epsilon_1/n}$, where $\epsilon_1$ can be chosen arbitrarily in $(0,1)$ such as $\epsilon_1=1/2$.

\item[Step 5:]
Measuring $|\widehat{\mv}_D\rangle$, we obtain a state $|\mathbf e_{k-1}\rangle$ which corresponds $m_{\bar D,k}$ in $\mv_D$.

\item[Step 6:]
Let $m_{\bar D,k}=\prod_{i=1}^{u_k} x_{n_i}$.
Set $x_{n_i} = 1$ in $\FS_1\subset\C[\Y]$ for $i=1,\ldots,u_k$.

\item[Step 7:]
Remove $0$ from $\FS_1$. Set $\Y=\Y\setminus\{x_{n_i}\,|\,i=1,\ldots,u_k\}$.

\item[Step 8:]
If $1\in\FS_1$ or $\Y=\emptyset$ then goto {\bf Step 11}

\item[Step 9:]
If $\FS_1\ne\emptyset$ and $\FS_1(\mathbf 0)\ne\mathbf 0$,
 then goto {\bf Step 3}.

\item[Step 10:]
Return $(a_1,\ldots,a_n)$ where $a_i=0$ if $x_i\in\Y$ else $a_i=1$.

\item[Step 11:]
If $l>\lceil\log_{\epsilon_1}\epsilon\rceil$ then return $\emptyset$, else $l=l+1$ and goto {\bf Step 2}.
\end{description}

We have the following theorem, which will be proved in the rest of this section.
\begin{thm}\label{th-eq}
Algorithm \ref{alg-b} has the following properties.
\begin{enumerate}
\item If the algorithm returns a solution, then it is a Boolean solution of $\FS=0$.
Equivalently, if $\FS$ has no Boolean solutions, the algorithm returns $\emptyset$.
\item If $\FS$ has Boolean solutions, the algorithm computes one with probability
at least $1-\epsilon$.
\item The runtime complexity of the algorithm is
$\widetilde O(n^{2.5}(n+T_\FS)\kappa^2\log1/\epsilon)$,
where $\kappa$ is the maximal condition number for all matrixes $\MD_{\FS_2,D}$
in Step 4 of the algorithm, called the {\em condition number}
for the polynomial system $\FS$.
\end{enumerate}
\end{thm}
%\begin{cor}
%If we use Ambainis' algorithm \cite{hhl-new} to solve linear system instead of HHL %algorithm, the complexity will be
%$$\widetilde O(n^{3.5}(n+rt)\kappa\log1/\epsilon).$$
%\end{cor}
%
%In the rest of this section, we will explain the Algorithm \ref{alg-b} and prove the correctness of Theorem \ref{th-eq}.
%

First, we briefly explain  {Algorithm \ref{alg-b}}.
The algorithm has two loops: the inner loop from Step 3 to Step 9
and the outer loop from Step 2 to Step 11.

In the inner loop, we try to find a solution $\mathbf a=(a_1,\ldots,a_n)$ of $\FS_1$ and in each run of the loop at least one coordinate of $\mathbf a$ whose value is $1$, say $a_k=1$, is found. Then, we set $x_k = 1$ and try to find the rest coordinates
of $\mathbf a$ in the rest of the loop.
If the inner loop fails, then we restart from Step 2 and try to find another solution.

The purpose of the outer loop is two folds.
%
%The number of runtimes $\lceil\log_{\epsilon_1}\epsilon\rceil$ of the outer loop is controlled by the parameter $\epsilon_1$.
%
First, by using precision $\sqrt{\frac{\epsilon_1}{n}}$
instead of $\sqrt{\frac{\epsilon}{n}}$ in Step 4,
the algorithm uses less qubits.
%
%If we do not use this loop, we need to set $\epsilon_1=\epsilon$ in step 4 and the precision needed is $\sqrt{\frac{\epsilon}{n}}$.
%
%With the loops started in Step 11,
We can use a large value for $\epsilon_1\in(0,1)$, say $\epsilon_1=1/2$.
Then the precision needed in Step 4 is $\sqrt{\epsilon_1/n}=\frac{1}{\sqrt{2n}}$ which
is generally larger than $\sqrt{\frac{\epsilon}{n}}$.
Second, the complexity  of the algorithm related with the precision decreases
from $O(1/\sqrt{\epsilon})$ to $O(\log{1/\epsilon})$.
The reason is the algorithm runs  $\lceil\log_{\epsilon_1}\epsilon\rceil$
more times of the inner loop but with less precision $\sqrt{\frac{\epsilon_1}{n}}$.
Please refer to Lemma \ref{lem-ec} for the detailed analysis.

We will explain each step of the algorithm below.
In Step 1, we first check two easy solutions.
In Step 2, since $x_i^2-x_i=0$, we replace $x_i^m$ by $x_i$ in time $\widetilde{O}(nT_\FS)$
to obtain $\FS_B$.
%As a consequence, $\deg(\FS)\le n$.

In the inner loop from step 3 to step 9, we will try to find a solution $\mathbf a=(a_1,\ldots,a_n)$ of $\FS_1$.
In Step 3, the bound $D=3\#\Y$ for the completely solving degree can be used due to Lemma \ref{lm-la1}.
In Step 4, we use the modified HHL to solve the Macaulay linear system.

In Step 5, we measure the quantum state $|\widehat{\mv}_D\rangle$.
Let $|\widehat{\mv}_D\rangle=
(\widehat{m}_{\bar D,1},\ldots,\widehat{m}_{\bar D,({\bar D}+1)^n-1})^T.$
Then by the property of quantum measurement,
with probability $|\widehat{m}_{\bar D,k}|^2$,
the measurement returns $|\mathbf e_{k-1}\rangle$
(By Remark \ref{rem-sb}, the subscript starts at $0$).

In Step 6, we will show later that with high provability
$m_{\bar D,k}=\prod_{i=1}^{u_k} x_{n_i}\ne0$ at the solution $\mathbf a=(a_1,\ldots,a_n)$ to be found.
Since $a_i$ is either $0$ or $1$,
$\prod_{i=1}^{u_k} a_{n_i}\ne0$ implies $a_{n_i}=1$ for all $n_i$.
We thus set $x_{n_i}=1$ in Step 7 and try to find the other coordinates of $\bf a$ in the loop from Step 3 to Step 9.

In Step 8, either $1\in\FS_1$ or $\Y=\emptyset$ implies $\V_{\C}(\FS_1,\H_\Y)=\emptyset$, because
we have both $\FS(\mathbf 0)\ne\mathbf 0$ and $\FS(\mathbf 1)\ne\mathbf 0$ from step 1.
This means that we did not find a solution
in the loop from Step 3 to Step 9
and need to find another solution by starting from Step 2 again.

In Step 9, if $\FS_1=\emptyset$ or $\FS_1(\mathbf 0)=\mathbf 0$,
then we find a solution of $\FS$:  $x_i=0$ for any $x_i\in\Y$
and $x_j=1$ for any $x_j\not\in\Y$,
which will be returned in Step 10.

%The purpose of the loop started in Step 11 is try to use less qubits.
%%
%If we do not use this loop, we need to set $\epsilon_1=\epsilon$ in step 4 and the precision needed is $\sqrt{\frac{\epsilon}{n}}$.
%%
%With the loops started in Step 11,
%we can use a large value for $\epsilon_1\in(0,1)$, say $\epsilon_1=1/2$,
%then the precision needed in Step 4 is $\sqrt{\epsilon_1/n}=\frac{1}{\sqrt{2n}}$ which
%is generally larger than $\sqrt{\frac{\epsilon}{n}}$.
%%
%%Otherwise, we repeat the procedure for $\FS_1$ in Step 11.
%%
%The cost for this is a factor of $\log_{\epsilon_1}\epsilon$
%in the complexity.

Secondly, we prove the correctness of Theorem \ref{th-eq}.
Part 1 of Theorem \ref{th-eq} is obviously true, since
we have checked this fact in the algorithm.
Part 2 of Theorem \ref{th-eq} follows
from Lemma \ref{lm-bc3}, and part 3 of  Theorem \ref{th-eq} follows
from Lemma \ref{lem-ec}.

The following key lemma gives the successful probability for Steps 5 and 6.
%
%In {Step 5}, the measurement of state $|\widehat{\mv}_D\rangle$ is different from the measurement of state $|{\mv}_D\rangle$.
%

\begin{lem}\label{lm-bc1}
In Steps 5 and 6, with a probability $>1-\epsilon_1/n$, $\V_{\C}(\FS_2)\ne\emptyset$ implies that there exists an $\mathbf a=(a_1,\ldots,a_n)\in\V_{\C}(\FS_2)$
with ${ a}_{n_i}=1$ for $i=1,\ldots,u_k$.
%
%or equivalently, if $\V(\FS_2)\ne\emptyset$ then
%$V(\FS_2,\{y_{n_i}+1\,|\,i=1,\cdots,u_k\})\ne\emptyset$.
\end{lem}
\begin{proof}
Let $|\mv_D\rangle=\sum_{j=1}^{(\bar D+1)^n-1} \alpha_j|\mathbf e_{j-1}\rangle$ be the solution state and $|\widehat{\mv}_D\rangle=\sum_{j=1}^{(\bar D+1)^n-1} \beta_j|\mathbf e_{j-1}\rangle$ be the approximate state obtained with the HHL algorithm.
If we can measure the true solution $|\mv_D\rangle$
and obtain $|\mathbf e_{k-1}\rangle$, then $\alpha_k\ne0$.
But the HHL algorithm actually returns $|\widehat{\mv}_D\rangle$.
By the definition of quantum measurement,
measuring $|\widehat{\mv}_D\rangle$ will return
$|\mathbf e_{k-1}\rangle$ with probability $\beta_k$.
Measuring $|\widehat{\mv}_D\rangle$ may lead a wrong
$|\mathbf e_{k-1}\rangle$, that is, $\alpha_k=0$ but $\beta_k\ne0$.
%We will give the probability for this wrong case to happen.
%
By the definition of quantum measurement, the probability for this wrong case to happen is $\|\sum_{j, \alpha_j=0}\beta_j|\mathbf e_{j-1}\rangle\|^2
=\|\sum_{j, \alpha_j=0} (\beta_j-\alpha_j)|\mathbf e_{j-1}\rangle\|^2
<\|\sum_{j=1}^{(\bar D+1)^n-1}(\beta_j-\alpha_j)|\mathbf e_{j-1}\rangle\|^2=\||\widehat{\mv}_D\rangle-|\mv_D\rangle\|^2<\epsilon_1/n$.
In other words, if the HHL algorithm returns $|\mathbf e_{k-1}\rangle$,
then with probability $> 1-\epsilon_1/n$,
the measurement returns a correct $|\mathbf e_{k-1}\rangle$£¬
meaning  $\alpha_k\ne0$.

By Theorem \ref{th-qsh}, the HHL algorithm returns
$\widehat{\mv}_D  = \sum\limits_{\mathbf a \in \V_\C(\FS_2)} \eta_{\mathbf a} \widetilde{\mv}(\mathbf a)$.
Then, we have $\alpha_k =$ $ \sum\limits_{\mathbf a \in \V_\C(\FS_2)} \eta_{\mathbf a}m_{\bar D,k} (\mathbf a)$.
The condition $\alpha_k\ne0$ implies that there exists a solution
${\mathbf a \in \V_\C(\FS_2)}$ such that
$m_{\bar D,k}(\mathbf a)\ne0$.
Since $\mathbf a$ is a Boolean solution,  we have $m_{\bar D,k}(\mathbf a)=1$. The lemma is proved.
\end{proof}

We now  compute the successful probability for the inner loop.
\begin{lem}\label{lm-bc2}
The loop from Step 3 to Step 9 will run at most $n$ times,
and returns $\emptyset$ with probability $<\epsilon_1$ when $\FS=0$ has Boolean solutions.
%
%The complexity for Step 1 to 10 is $\widetilde O(n^{2.5}(n+T)\kappa^2\epsilon_1^{-0.5})$.
\end{lem}
\begin{proof}
Since at each loop, the values of at least one $x_i$ will be determined
in Step 6,
we will repeat this loop for at most $n$ times.
By Lemma \ref{lm-bc1}, when $\FS=0$ has Boolean solutions, the algorithm returns $\emptyset$  with probability $< 1-(1-\epsilon_1/n)^{n}<\epsilon_1$.
\end{proof}

We now  compute the successful probability for the algorithm.
\begin{lem}\label{lm-bc3}
The loop from Step 2 to Step 11 will run at most $\lceil\log_{\epsilon_1}\epsilon\rceil$ times and with probability $\ge 1-\epsilon$, returns a Boolean solution of $\FS=0$ when $\FS=0$ has Boolean solutions.
\end{lem}
\begin{proof}
By Lemma \ref{lm-bc2},  if  $\FS$ has Boolean solutions, then
the probability that we reach step 11 is $<\epsilon_1$.
The number of loops from Step 2 to Step 11 is
at most $\lceil\log_{\epsilon_1}\epsilon\rceil$.
Then, if  $\FS$ has Boolean solutions, then
the probability that the algorithm returns $\emptyset$ is
$\epsilon_1^{\lceil\log_{\epsilon_1}\epsilon\rceil}<\epsilon$.
\end{proof}

Finally, we estimate the runtime complexity of Algorithm \ref{alg-b}.
\begin{lem}\label{lem-ec}
The complexity for {Algorithm \ref{alg-b}} is
$\sqrt{2}c(n\log_2(6n+1) +\log_2(r+1))n^{1.5}(n+1+T_\FS)\kappa^2\lceil\log_{2}1/\epsilon\rceil$.
Moreover, it equals to
$\widetilde O(n^{2.5}(n+T_\FS)\kappa^2\log1/\epsilon)$.
\end{lem}
\begin{proof}
Step 4 is the dominate step in terms of complexities.
The complexities for other steps are very low comparing to that of
Step 4. So, we just omit them from the complexity analysis.

We have $D\le2n+d$ from {Lemma \ref{sol-sh}}.
Due to Step 2, we have $d\le n$, so $D\le 3n$.
$\MD_{\FS_2,D}$ is of dimension $(r(\bar{d}+1)^n)\times ((\bar D+1)^n -1)$ and $(2n+T_\FS)$-sparseness.
%By {Theorem \ref{th-qsh}},
By Corollary \ref{cor-E},
the complexity of Step 4 is approximately $c\log(r (\bar d+1)^n+(\bar D+1)^n -1)(2n+T_\FS)\kappa^2\sqrt{n/\epsilon_1}$.

By Lemma \ref{lm-bc2}, the loop from Step 3 to Step 9 will run at most $n$ times. Then the complexity for the loop from Step 3 to Step 9 is $\sum_{j=0}^{n-1}(c\log(r (\bar d+1)^n+(\bar D+1)^n  -1)(2(n-j)+T_\FS)\kappa^2\sqrt{n/\epsilon_1})\le c\log(r (\bar d+1)^n+(\bar D+1)^n -1)(n(n+1)+nT_\FS)\kappa^2\sqrt{n/\epsilon_1}$.

By Lemma \ref{lm-bc3}, the loop from Step 2 to Step 11 will run at most
$\lceil\log_{\epsilon_1}\epsilon\rceil$ times.
Then the total complexity of the algorithm is
$c\log(r (\bar d+1)^n+(\bar D+1)^n -1)(n(n+1)+nT_\FS)\kappa^2\sqrt{n/\epsilon_1}
\lceil\log_{\epsilon_1}\epsilon\rceil=
c\log(r (\bar d+1)^n+(\bar D+1)^n -1)n^{1.5}(n+1+T_\FS)\kappa^2\sqrt{2}\lceil\log_{2}1/
\epsilon\rceil$, by choosing $\epsilon_1$ to be $1/2$.

Since $r (\bar d+1)^n+(\bar D+1)^n -1\le (r+1)(\bar D+1)^n$,
 we have $\log(r (\bar d+1)^n+(\bar D+1)^n -1)\le\log(r+1)+n\log(\bar D+1)\le\log(r+1)+n\log(2D+1)\le\log(r+1)+n\log(6n+1)$ by {Lemma  \ref{lem-6d} ($\bar D+1\le2D+1)$ and \ref{lm-la1} ($D\le3n$)}.

The totally  complexity for {Algorithm \ref{alg-b}} is at most
$\sqrt{2}c(n\log_2(6n+1) +\log_2(r+1))n^{1.5}(n+1+T_\FS)\kappa^2
\lceil\log_{2}$ $1/\epsilon\rceil=\widetilde O((n+\log(r))n^{1.5}(n+T_\FS)\kappa^2\log$ $1/\epsilon)
=\widetilde O(n^{2.5}(n+T_\FS)\kappa^2\log1/\epsilon)$, because $r\le T_\FS$.
\end{proof}

We have completed the proof of {Theorem \ref{th-eq}}.
We can easily improve our algorithm as follows.
\begin{rem}
Given $(a_1,\ldots,a_n)\in\C^n$, $\FS\subset\C[\X]$, we can obtain an element in $\V_\C(\FS,x_1^2-a_1x_1,\ldots,x_n^2-a_nx_n)$ by {Algorithm \ref{alg-b}}, where we need to replace $\H_\X$ with $(x_1^2-a_1x_1,\ldots,x_n^2-a_nx_n)$ and $x_{n_i}=1$ with $x_{n_i}=a_{n_i}$ in Step 6.
\end{rem}

\begin{rem}
Algorithm \ref{alg-b} has complexity
$\widetilde O(n^{3.5}(n+T_\FS)\kappa\log1/\epsilon)$ if using
Ambainis' algorithm \cite{hhl-new} to solve the Macaulay linear system.
\end{rem}

\subsection{Obtain all the Boolean solutions}
We will show how to find all Booelan solutions of $\FS$.
For a Boolean solution $\mathbf a$ of $\FS$,
the following lemma shows how to construct a polynomial system $\FS_1$
satisfying $\V_{\C}(\FS_1,\H_\X) = \V_{\C}(\FS,\H_\X)\setminus\{\mathbf a\}$.

\begin{lem}
For $\mathbf a=(a_1,\ldots,a_{n})\in\V_{\C}(\FS,\H_\X)$, we have
$$\proj_\X\V_{\C}(\FS,\H_\X,\NS,f_{\mathbf a})=\V_{\C}(\FS,\H_\X)\setminus\{\mathbf a\}$$
where $\NS = \{\bar x_i+x_i-1\,|\,i=1,\ldots,n\}$,
$f_{\mathbf a}=\prod_{a_i=0}\bar x_i\prod_{a_i=1}x_i$, and $\bar x_i$ are new variables.
\end{lem}

Then we can use the Algorithm \ref{alg-b} to find all Boolean Solutions for $\FS=0$.
\begin{alg}\label{alg-a}
\end{alg}

{\noindent\bf Input:}
$\FS=\{f_1,\ldots,f_r\}\subset \C[\X]$ with $T_\FS=\sum_{i=1}^r\#f_i$ and $d_i =\deg(f_i)$. Also $\epsilon\in(0,1)$.

{\noindent\bf Output:}
 $\V_{\C}(\FS,\H_\X)$.

\begin{description}

\item[Step 1:]
Set $S=\emptyset$. $\FS_1=\FS\cup\{x_1+\bar{x}_1-1,\ldots,x_n+\bar{x}_n-1\}\subset\C[\X,\overline{\X}]$
with $\overline{\X}=\{\bar{x}_1,\ldots, \bar{x}_n\}$.

\item[Step 2:]
Use {Algorithm \ref{alg-b}} to compute Boolean solutions of $\FS_1=0$.
If we obtain $\emptyset$, return $S$.
Else we obtain a Boolean solution $\mathbf a=(a_1,\ldots,a_n)$.

\item[Step 3:]
$S=S\cup\{\mathbf a\}$, $\FS_1=\FS_1\cup\{\prod_{a_i=0}\bar x_i\prod_{a_i=1}x_i\}$. Goto {Step 2}.

\end{description}

\begin{thm}
Let $w=\#\V_{\C}(\FS,\H_\X)$.
Then {Algorithm \ref{alg-a}} finds $w$ Boolean solutions of $\FS=0$ with  complexity $\widetilde O(n^{2.5}(n+T_\FS+w)w\kappa^2\log1/\epsilon)$, and probability at least $(1-\epsilon)^{w}$.
%, where $w=\#\V_{\C}(\FS,\H_\X)$.
\end{thm}
\begin{proof}
By Theorem \ref{th-eq}, the complexity of the algorithm  is $\sum_{i=0}^{w-1}\widetilde O((2n)^{2.5}(2n+3n+T_\FS+i)\kappa^2\log1/\epsilon)$ $=\widetilde O(n^{2.5}(n+T_\FS+w)w\kappa^2\log1/\epsilon)$.
\end{proof}
%
%\begin{rem}
%When setting $x_i=1$ in Step 6 of {Algorithm \ref{alg-b}}, we add $\bar x_i=0$ to half the time.
%\end{rem}

\subsection{Computing Boolean solutions to linear systems and applications}
\label{sec-lin}
Many well-known problems in computation theory and cryptography can be described
as finding the Boolean solutions for linear systems over $\C$.
In this section, we consider two such problems and their computational complexities
using our quantum algorithm.
%
%In this section, we will apply our algorithm to the subset sum problem and the graph isomorphism problem.
%
%\subsection{Knapsack problem}

The {\em subset sum problem} is an important problem in complexity theory and cryptography. The problem is: given a set of integers, is there a non-empty subset whose sum is a given number?
The problem can be described as finding the Boolean solutions
for a linear system.
We have the following result.
\begin{prop}\label{p-1}
Let $A\in \Z^{r\times n}$ for $r < n$ and $b\in \Z^r$. There is a quantum algorithm to find a Boolean solution to the linear system $Ax=b$ with probability
$\ge 1-\epsilon$ and complexity $\widetilde O(n^{3.5}r\kappa^2\log1/\epsilon)$
\end{prop}
\begin{proof}
We have $r$ linear equation of sparseness $(n+1)$ and $n$ quadratic binomials. Thus $T=2n+nr+r$, and by {Theorem \ref{th-eq}}, we can use {Algorithm \ref{alg-b}} to find a Boolean solution for $Ax=b$ in time $\widetilde O(n^{2.5}(n+2n+nr+r)\kappa^2\log1/\epsilon)=\widetilde O(n^{3.5}r\kappa^2\log1/\epsilon)$.
\end{proof}

%\subsection{Graph isomorphism problem}
The {\em graph isomorphism problem} is another well-known problem in computational theory, which is to determine whether two finite graphs are isomorphic.
We do not know whether it is NPC or P.
The problem can be described as solving the Boolean solutions for a linear
system.
Let $A$ and $B$ in $\F_2^{n\times n}$ be the adjacent matrices for two graphs, the graph isomorphism problem is  to decide whether there exists a permutation matrix $P$ such that $AP=PB$.

\begin{prop}\label{p-2}
There is a quantum algorithm
to decide whether two graphs with $n$ vertices are isomorphic with probability $\ge 1-\epsilon$ and complexity $\widetilde O(n^{6.5}\kappa^2\log1/\epsilon)$.
\end{prop}
\begin{proof}
Let $A=(a_{ij})$, $B=(b_{ij})$, $P=(x_{ij})$ with $\sum_ix_{ij}=1$ for each $j$, $\sum_jx_{ij}=1$ for each $i$, and $x_{ij}^2-x_{ij}=0$ for each $i,j$.
Thus in the equation system, the number of $2n$-sparse linear equations is $n^2$,
the number of $(n+1)$-sparse linear equations is $2n$, and the number of quadratic binomials is $n^2$.
Thus $T=2n^3+4n^2+2n$, by Theorem \ref{th-eq}, we can use {Algorithm \ref{alg-b}} to find a Boolean solution for $AP=PB$ in time $\widetilde O((n^2)^{2.5}(n^2+2n^3+4n^2+2n)\kappa^2\log1/\epsilon)
=\widetilde O(n^{8}\kappa^2\log1/\epsilon)$.

Due to the special property of the problem, the complexity could be reduced as follows. Considering the loop from Step 3 to Step 9 in {Algorithm \ref{alg-b}}, since exactly $n$ of $x_{ij}$ equal to $1$ in the permutation matrix $P$, the number of loops will be $n$ instead of $n^2$. Thus the error bound in step 4 will be $\sqrt{\epsilon_1/n}$ instead of  $ \sqrt{\epsilon_1/n^2}$. Finally, we can use {Algorithm \ref{alg-b}} to find a Boolean solution for $AP=PB$ in time $\widetilde O(n^{8-1.5}\kappa^2\log1/\epsilon)=\widetilde O(n^{6.5}\kappa^2\log1/\epsilon)$.
\end{proof}

By Propositions \ref{p-1} and \ref{p-2}, in order to determine the quantum
complexity of these two problems, we need only to study the condition
numbers of the corresponding equation systems.

\section{Solving Boolean equation systems}
\label{sec-be}
In this section, we will give a quantum algorithm
to solve Boolean equations by converting the problem into
that of computing the Boolean solutions for a $6$-sparse polynomial system over $\C$.

\subsection{Reduce Boolean systems to polynomial systems over $\C$}

Let $\mathbb F_2$ be the field consisting of $0$ and $1$.
We will consider the problem of equation solving over $\mathbb F_2$,
or equivalently, solving Boolean equations.
Let $\X=\{x_1,\ldots,x_n\}$ be a set of indeterminants and
$$\mathcal R_2[\X]=\mathbb F_2[\X]/(\mathbb H_{\X}),$$ where $\mathbb H_{\X}=\{x_1^2-x_1,\ldots,x_n^2-x_n\}$.
%and we don't distinguish $\mathbb H_\X$ over $\C$ or $\mathbb F_2$.
Then $\mathcal R_2[\X]$ is a Boolean ring
and every ideal in $\mathcal R_2$ is radical.
Elements in $\mathcal R_2$ are called {\em Boolean polynomials},
which have the form $\sum_{i} m_{i}$ and $m_i$ are Boolean monomials
with degree at most one for each $x_i$.
Similar to Section 3, we use $\V_{\FB_2}(\FS)$ to denote the zeros
of $\FS\subset\mathcal R_2[\X]$ in $\FB_2$.

The following example shows that we cannot use
a method of equation solving over $\C$ to solve Boolean equations directly.
\begin{exmp}\label{ex-2}
Let $f=x_1+x_2+1$. Then $\V_{\FB_2}(f) = \{(0,1),(1,0)\}$.
But $\V_{\C}(f,x_1^2-x_1,x_2^2-x_2) = \emptyset$.
\end{exmp}
The following lemma shows how to reduce Boolean equation solving to equation solving over $\C$.
\begin{lem}\label{lm-s0}
Let $\FS=\{f_1,\ldots,f_r\}$ be a set of Boolean polynomials with $t_i=\# f_i$.
In $\C[\X]$, let $C(f_i)=\prod_{k=f_i(\mathbf 0)}^{\lfloor t_i/2\rfloor} (f_{i}-2k)$ and let
\begin{equation}\label{eq-C}
C(\FS) = \{C(f_1),\ldots,C(f_r)\}\subset \C[\X].
\end{equation}
%
% \mod(HS)
%
Then $\V_{\FB_2}(\FS)=\V_{B}(C(\FS))$ which is the Boolean solutions of $C(\FS)$.
Furthermore, $\#C(f_i)\le t_i(t_i+1)^{\lfloor t_i/2\rfloor}$.
%
%$(C(\FS))$ is a radical ideal in $\C[\X]$ and $C(\FS)$ is of form \bref{eq-FSB}.
\end{lem}
\begin{proof}
Let $f_i=\sum_{k=1}^{t_i} m_{ik}$ and
 $\mathbf a=(a_1,\ldots,a_{n})\in \V_{\FB_2}(\FS)$. When we regard $f_i$ as a polynomial in $\mathbb C[\X]$, $f_i(\mathbf a)=\sum_{k=1}^{t_i} m_{ik}(\mathbf a)$ is an even integer between $f_i(\mathbf 0)$ and $t_i$, because $f_i(\mathbf a)\equiv 0\ \mod\ 2$. Thus $\mathbf a$ is a zero of $C(f_i)=\prod_{k=f_i(\mathbf 0)}^{\lfloor t_i/2\rfloor} (f_{i}-2k)$.
%
%Considering when $f_i(\mathbf 0)=1$, $1\le f_i(\mathbf a)\ne 0$, we delete the factor $f_i$ to get $F_i=\prod_{k=f_i(\mathbf 0)}^{\lfloor t_i/2\rfloor} (f_{i}-2k)$.
%
The other direction is easy: a Boolean solution $\mathbf a$ of $C(f_i)=\prod_{k=f_i(\mathbf 0)}^{\lfloor t_i/2\rfloor} (f_{i}-2k)$ satisfies
$f_i(\mathbf a)=2k$ for some $k$, and hence $\mathbf a$ is a Boolean solution for $\FS$.
%By {Lemma \ref{lm-la1}}, $(C(\FS))$ is radical. Clearly $C(\FS)$ of form \bref{eq-FSB}.
%
Finally, in the worst case, $f_i(0)=0$ and
$\#C(f_i)\le t_i(t_i+1)^{\lfloor t_i/2\rfloor}$.
\end{proof}

Note that $\#C(f_i)$ increases exponentially in terms of $t_i$.
In order to obtain sparse polynomials, we split each $f_i$
into several $s$-sparse Boolean polynomials for a given $s\in \N_{\ge3}$.
Let $f=\sum_{i=1}^{t} m_{i}$ be a Boolean polynomial.
Set
$S_t=\lceil\frac{t-s}{s-2}\rceil$ and $\U_f = \{u_1,\ldots,u_{S_t}\}$ be a set of new variables depending on $f$.
We define a Boolean polynomial set $S(f,s)$ as follows.
If $t\le s$, then $S(f,s) =\{f\}$ and $\U_f=\emptyset$. Otherwise, let
\begin{equation}\label{eq-S1}
S(f,s) = \{\check{f}_1,\ldots, \check{f}_{S_t+1}\} \subset \mathcal R_2[\X,\U_f]
\end{equation}
where
%$S_t=\lceil\frac{t-s}{s-2}\rceil$,
$\check{f}_1=\sum_{k=1}^{s-1} m_{k}+u_1$, $\check{f}_{j}=\sum_{k=(j-1)(s-2)+2}^{j(s-2)+1}m_{k}+u_{j-1}+u_{j}$
for $j=2,\ldots,S_t$,
and $\check{f}_{S_t+1}=\sum_{k=S_t(s-2)+2}^{t}m_{k}+u_{S_t}$.
$S(f,s)$ is called the {\em splitting set} for $f$.

For convenience of presentation, in this paper,
we give new meaning to the notation: $\lceil e\rceil=0$ if $e\le 0$.
With this assumption, $\#\U_f = \lceil\frac{t-s}{s-2}\rceil$
and $\#S(f,s) = \lceil\frac{t-s}{s-2}\rceil+1$.

For  a set $\FS$ of Boolean polynomials, denote $\U(\FS,s) = \bigcup_{f\in\FS}\U(f,s)$, and
\begin{equation}\label{eq-S}
S(\FS,s) =\bigcup_{f\in\FS} S(f,s)\subset \mathcal R_2[\X,\U(\FS,s)].
\end{equation}

%For the convenience of presentation, we assume $\#f_i \ge s$, without loss of generality. Because, if $\#f_i < s$.
%
The following results are easy to check.
\begin{lem}\label{lm-s1}
Let $\FS=\{f_1,\ldots,f_r\}\subset \mathcal R_2[\X]$ and $t_i=\#f_i$.
Then
$S(\FS,s)$ is $s$-sparse,
$\# S(\FS,s)=r+\sum_i\lceil\frac{t_i-s}{s-2}\rceil$,
$\#(\X\cup\U(\FS,s)) =n+\sum_i\lceil\frac{t_i-s}{s-2}\rceil$.
%
%In particular, $\# S(\FS,3)=T-2r$ and $\#(\X\cup\U(\FS,3))=n+T-3r$, where $T=\sum_i t_i$.
\end{lem}

\begin{lem}\label{lm-s2}
Let $\FS=\{f_1,\ldots,f_r\}\subset \mathcal R_2[\X]$.
For a given $s\in\N_{\ge3}$, we have
$$(S(\FS,s)) \bigcap \mathcal R_2[\X] = (\FS),\ \proj_\X\V_{\FB_2}(S(\FS,s))=\V_{\FB_2}(\FS),$$
where $(S(\FS,s))$ is the ideal generated by $S(\FS,s)$ in $\mathcal R_2[\X,\U(\FS,s)]$.
%
%Furthermore, $\#(S(\FS,s)) \ge s$,
%$\deg S(\FS,s) = \deg \FS$,
%$\# S(\FS,s) = r+\sum_i\lfloor\frac{t_i-1}{s-2}\rfloor$,
%$\# (\X\cup \U(\FS,s)) =n+\sum_i\lfloor\frac{t_i-1}{s-2}\rfloor$.
\end{lem}

%\begin{cor}\label{lm-cd}
%Let $F_i$ be defined in Lemma \ref{lm-s0}. Then the completely solving degree of $C(\FS)$ satisfies $\le (t/2+1)d+2n$, where $t=\max_i\#f_i$ and $d=\max_i\deg(f_i)$.
%\end{cor}
%\begin{proof}
%By {Lemma \ref{lm-csd}}, the completely solving degree of $(C(\FS))$ satisfies $\csdeg(C(\FS))\le$
%$\max_i\deg(F_i)$ $+2n\le(t/2+1)d+2n$.
%\end{proof}
%\subsection{Reduce Boolean equations to $s$-sparse equations}

We summarize the results of this subsection as the following theorem
which follows from Lemmas \ref{lm-s1} and \ref{lm-s0}.
\begin{thm}\label{th-p}
Let $\FS=\{f_1,\ldots,f_r\}\subset \mathcal R_2[\X]$ with $t_i=\#f_i$
and $d_i = \deg(f_i)$.
For a given $s\in\N_{\ge3}$, let $\Y = \X\cup\U(\FS,s)$.
Then we have a polynomial set
$$P(\FS,s)=C(S(\FS,s))\subset \C[\Y]$$
such that $\V_{\FB_2}(\FS) =
 \proj_{\X}
 \V_{B}(P(\FS,s))$.
Furthermore,
$P(\FS,s)$ is $s(s+1)^{\lfloor s/2\rfloor}$-sparse,
$\# P(\FS,s) = r+\sum_i\lceil\frac{t_i-s}{s-2}\rceil$,
$\# \Y =n+\sum_i\lceil\frac{t_i-s}{s-2}\rceil$.
\end{thm}
%\begin{proof}
%For any $\check F=\prod_{k=\check f(\mathbf 0)}^{\lfloor s/2\rfloor} (\check f-2k)\in %C(S(\FS,s))$, $\#\check f\le s$ implies $\#\check F\le s(s+1)^{\lfloor s/2\rfloor}$.
%, the completely solving degree of $P(\FS,s)$ is $\csdeg(P(\FS,s))\le(s/2+1)
%\deg(S(\FS,s))+2\#\Y=(s/2+1)d+2n+2
%\sum_i\lceil\frac{t_i-s}{s-2}\rceil<
%(s/2+1)d+2n+rt=O(sd+n+rt)$
%
%Other results are from Lemmas \ref{lm-s1} and \ref{lm-s0}.
%\end{proof}

In our algorithm, we use $s=3$.
For a Boolean polynomial $f$ with $\# f \le 3$, we give the following improved version of $C(f)$.
\begin{equation}\label{eq-3ss}
\widehat{C}(f)=\begin{cases}
f,& \text{if } \# f=1;\\
n_{1} - n_{2}, &\text{if } \# f=2 \text { and }f=n_{1} + n_{2};\\
f-2,& \text{if } \# f=3 \text{ and }f(\mathbf 0)=1;\\
2m_{1}m_{2}+2m_{1}m_{3}+2m_{2}m_3 - f & \text{if }\# f=3 \text{ and } f(\mathbf 0)=0,
\end{cases}
\end{equation}
where $f=m_{1}+m_{2}+m_{3}$  in the case $\#f = 3$.
For this new $\widehat{C}(f)$, we have
\begin{cor}\label{lm-6s}
For $s=3$,
$\V_{\FB_2}(\FS) =
 \proj_{\X}
 \V_{B}(P(\FS,s))$.
Furthermore,
$P(\FS,3)$ is $6$-sparse,
$\# P(\FS,3)\le T_\FS$,
$\#\Y \le T_\FS+n$.
\end{cor}
\begin{proof}
We need only to show the last case in \bref{eq-3ss} and other results are easy to verify.
Since we consider Boolean solutions, we may set $m^2=m$ for any monomial $m$.
Then, $C(f_i)=f_i(f_i-2)=m_{i1}^2-2m_{i1}+m_{i2}^2-2m_{i2}+m_{i3}^2-
2m_{i3}+2m_{i1}m_{i2}+2m_{i1}m_{i3}+2m_{i2}m_{i3}=
2m_{i1}m_{i2}+2m_{i1}m_{i3}+2m_{i2}m_3-m_{i1}-m_{i2}-m_{i3}\ (\text{mod}\ \H_\Y)= \widehat{C}(f_i) (\text{mod}\ \H_\Y)$,
and  $\widehat{C}(f_i)$ is $6$-sparse.
%Then $(F_1,\ldots,F_r,\H_\Y)=(\hat F_1,\ldots,\hat F_r,\H_\Y)$.

If $t_i\ge3$ for all $i$,
by Lemma \ref{lm-s1}, we have $\# S(\FS,3)=T_\FS-2r$ and $\#(\X\cup\U(\FS,3))=n+T_\FS-3r$, where $T_\FS=\sum_i t_i$.
Then, $\# P(\FS,3) = T_\FS-2r \le T_\FS$,
$\#\Y =T_\FS+n-3r\le T_\FS+n$.
If $t_i<3$, then $S(f_i,3)=f_i$ and $\U(f_i,3)=\emptyset$. Hence the bounds are still true.
\end{proof}

\subsection{Quantum algorithm for Boolean equation solving}
In this subsection, we will give a quantum algorithm to solve Boolean equations.

\begin{alg}\label{alg-q}
\end{alg}

{\noindent\bf Input:}
$\FS=\{f_1,\ldots,f_r\}\subset \mathcal R_2[\X]$ and $\epsilon\in(0,1)$.
%with $t_i=\#f_i$ and $d_i =\deg(f_i)$. Also $\epsilon\in(0,1)$.

{\noindent\bf Output:}
A zero of $\FS$ or $\emptyset$ meaning that $\V_{\F_2}(\FS) = \emptyset$
 with probability $> 1-\epsilon$.

\begin{description}
\item[Step 1:]
Compute $\FS_1 = S(\FS,3)\subset\mathcal R_2[\Y]$ as defined in \bref{eq-S}, where $\Y = \X\cup \U(\FS,3)$.

\item[Step 2:]
Compute $\FS_2 = \widehat{C}(\FS_1)\subset\C[\Y]$ as defined in \bref{eq-3ss}.

\item[Step 3:]
Use {Algorithm \ref{alg-b}} to find a Boolean solution of $\FS_2=0$ over $\C[\Y]$, with the probability bound $\epsilon$. {Return} $\emptyset$ if {Algorithm \ref{alg-b}} returns $\emptyset$, else we have a Boolean solution $\mathbf a$ for $\FS_2$.

\item[Step 4:]
{Return} $\proj_\X\mathbf a$.

\end{description}

\begin{thm}\label{th-q}
Algorithm \ref{alg-q} has the following properties.
\begin{itemize}
\item If the algorithm returns a solution, then it is a solution of $\FS=0$.
Equivalently, if $\V_{\F_2}(\FS)=\emptyset$, the algorithm returns $\emptyset$.
\item If $\V_{\F_2}(\FS)\ne\emptyset$, the algorithm computes a solution of $\FS=0$ with probability $>1-\epsilon$.
\item The runtime complexity is $\widetilde O((n^{3.5}+T_\FS^{3.5})\kappa^2\log1/\epsilon)$,
where $T_\FS=\sum_i \#f_i$ and
$\kappa$ is the   condition number of
the polynomial system $\FS_2$, called the {\em condition number} of
the Booolean system $\FS$.
\end{itemize}
\end{thm}
\begin{proof}
In Step 2, we split $\FS$ to $3$-sparse polynomials and then turn them into a polynomial system over $\C$ in time $O(T_\FS)$. By Corollary \ref{lm-6s}, $\V_{\F_2}(\FS)=\V_{B}(\FS_2)$.
Thus,  we need only to find Boolean solutions for $\FS_2$ in Step 3.

%The following lemmas gives the complexities of Step 3.
If we use $\FS_1=S(\FS,s)$ instead of $\FS_1=S(\FS,3)$ in Step 1, by {Theorems \ref{th-eq} and \ref{th-p}}, the complexity of Step 3 is $\widetilde O((n+\sum_i\lceil\frac{t_i-s}{s-2}\rceil)^{2.5}((n+\sum_i\lceil\frac{t_i-s}{s-2}\rceil)+
(r+\sum_i\lceil\frac{t_i-s}{s-2}\rceil))s(s+1)^{\lfloor s/2\rfloor})\kappa^2\log1/\epsilon)=(n+T_\FS/s)^{2.5}(n+T_\FS s^{s/2})\kappa^2\log1/\epsilon)$. %
To minimize this complexity, we choose $s=3$, meanwhile the complexity is $\widetilde O((n+T_\FS)^{2.5}(n+T_\FS)\kappa^2\log1/\epsilon)=\widetilde O((n^{3.5}+T_\FS^{3.5})\kappa^2\log1/\epsilon)$.
\end{proof}

\begin{rem}
%In Step 8 of {Algorithm \ref{alg-b}}, we can replace $\Y_{\ref{alg-b}}=\emptyset$ with $\X_{\ref{alg-q}}\cap\Y_{\ref{alg-b}}=\emptyset$ so that  {Algorithm \ref{alg-q}} can terminate early.
Algorithm \ref{alg-q} has complexity
$\widetilde O((n^{4.5}+T_\FS^{4.5})\kappa\log1/\epsilon)$ if using Ambainis' algorithm \cite{hhl-new} to solve the Macualay linear systems.
\end{rem}

%\begin{exmp}
%$n=5$, $\FS=\{x_1+x_3+x_5,x_2+x_3+1,x_1x_4+x_2+1,x_3x_5\}$.
%When we use {Algorithm \ref{alg-b}} to solve $\V(\FS,\H_\X)$, we have the $\sdeg(\FS\cup\H_\X)=8$, the modified Macaulay matrix is of dimension $3894\times1286$, and its condition number is $\kappa=15.44$.
%\end{exmp}
%
%
%The following lemma shows that we can find Boolean solutions over certain finite fields.
%\begin{cor}
%For an equation system $\FS=\{f_1,\ldots,f_r\}$ over a finite field $\F_p$ for a prime number $p$, if all solution of $\FS$ is Boolean $(0$ or $1)$, we can also use {Algorithm \ref{alg-q}} to compute a solution of $\FS$ by replacing $F_i=\prod_{k=f_i(\mathbf 0)}^{\lfloor t_i/2\rfloor} (f_{i}-2k)$ with $F_i=\prod_{k=f_i(\mathbf 0)}^{\lfloor t_i/p\rfloor} (f_{i}-pk)$ in Step 2.
%\end{cor}

The following theorem gives the exact complexity for
solving Boolean equations, which will be used in Section \ref{sec-ca}.
\begin{thm}\label{cor-ec}
Let $\FS=\bigcup_{s=1}^t\{f_{s1},\ldots,f_{sr_s}\}\subset\mathcal R_2[\X]$ be a Boolean equation system such that $s=\#f_{sj}$, $T_\FS=\sum_{s=1}^tsr_s$, $r=\sum_{s=1}^tr_s$. Then we can find a solution of $\FS=0$ with runtime $\sqrt{2}c((n+T_\FS-3r+2r_1+r_2)\log_2(6(n+T_\FS-3r+2r_1+r_2)+1) +\log_2(T_\FS-2r+2r_1+r_2+1))(n+T_\FS-3r+2r_1+r_2)^{1.5}
((n+T_\FS-3r+2r_1+r_2)+1+(6T_\FS-12r+7r_1+2r_2))\kappa^2\lceil
\log_{2}1/\epsilon\rceil\le
\sqrt{2}c(\log_2(n+T_\FS)+3) (n+T_\FS)^{2.5}
(n+7T_\FS)\kappa^2\lceil
\log_{2}1/\epsilon\rceil
$,
where $c$ is the complexity constant of the HHL algorithm
defined in Remark \ref{rem-HHL}.
\end{thm}
\begin{proof}
We use $\widehat{C}(f)$ defined in \bref{eq-3ss}. By Corollary \ref{lm-6s}, $C(S(\FS,3))$ consists of $r_1$ monomials, $r_2$ binomials, and $\sum_{s=3}^t((s-2)r_s)=T_\FS-2r+r_1$ polynomials of sparseness $6$, and the number of indeterminates is $n+\sum_{s=3}^t((s-3)r_s)=n+T_\FS-3r+2r_1+r_2$. Thus the total sparseness for $P=C(S(\FS,3))$ is $T_P=r_1+2r_2+6(T_\FS-2r+r_1)=6T_\FS-12r+7r_1+2r_2$.
By {Lemma \ref{lem-ec}}, the exact complexity for {Algorithm \ref{alg-q}} to find a solution is
$\sqrt{2}c((n+T_\FS-3r+2r_1+r_2)\log_2(6(n+T_\FS-3r+2r_1+r_2)+1) +\log_2(T_\FS-2r+2r_1+r_2+1))(n+T_\FS-3r+2r_1+r_2)^{1.5}
((n+T_\FS-3r+2r_1+r_2)+1+(6T_\FS-12r+7r_1+2r_2))\kappa^2\lceil
\log_{2}1/\epsilon\rceil
\le\sqrt{2}c((n+T_\FS)\log_2(6(n+T_\FS)+1) +\log_2(T_\FS))(n+T_\FS)^{1.5}
((n+T_\FS)+6T_\FS)\kappa^2\lceil
\log_{2}1/\epsilon\rceil\le
\sqrt{2}c(\log_2(n+T_\FS)+3) (n+T_\FS)^{2.5}
(n+7T_\FS)\kappa^2\lceil
\log_{2}1/\epsilon\rceil$.
\end{proof}

\subsection{Application to 3-satisfiability problem}
\label{sec-3sat}
Let $\X=\{x_1,\ldots,x_n\}$ be Boolean indeterminates. A 3-SAT problem is to
check the satisfiability  of the propositional logic formula $y_{i1}\vee y_{i2}\vee y_{i3}=1$ for $i=1,\ldots,r$,
where $y_{ij}=x_k$ or $\neg x_k$ for some $k$. Decision of 3-SAT is NPC.
The 3-SAT problem  is equivalent to solve the  Boolean equation system
$$\FS=\{\bar y_{i1}\bar y_{i2}\bar y_{i3}:i=1,\ldots,r\}\cup\{x_k+\bar x_k+1:k=1,\ldots,n\}\}$$
in $\R_2[\X,\overline\X]$, where $\overline\X=\{\bar x_1,\ldots,\bar x_n\}$.
\begin{prop}
For a 3-SAT with $r$ clauses, there is a quantum algorithm
to decide its satisfiability with probability $\ge 1-\epsilon$ and with complexity
$\widetilde O((n^{2.5}(n+r)\kappa^2\log1/\epsilon)$.
\end{prop}
\begin{proof}
It is easy to see that solving the Boolean system $\FS$ is equivalent to
find the Boolean solutions for the following polynomial system in $\C[\X,\overline\X]$,
$$\FS_1=\{\bar y_{i1}\bar y_{i2}\bar y_{i3}:i=1,\ldots,r\}\cup\{x_k+\bar x_k-1:k=1,\ldots,n\} \cup\{x_k^2-x_k:k=1,\ldots,n\},$$
that is, the 3-SAT problem is satisfiable if and only if $\V_{\C}(\FS)\ne\emptyset$.
Also note that $x+\bar x-1=0$ and $x^2-x=0$ imply $\bar x^2-\bar x=0$.
%
%$\FS$ has $2n$ indeterminates, $r$ cubic monomials, $n$ quadratic binomials, and $n$ linear trinomials.
%
$\FS_1$  consists of $n$ binomials, $n$ trinomials and $r$ monomials. Thus $T_{\FS_1}=5n+r$, by {Lemma \ref{lem-ec}}, we can use {Algorithm \ref{alg-b}} to find a Boolean solution in time $\widetilde O(n^{2.5}(n+5n+r)\kappa^2\log1/\epsilon)=\widetilde O(n^{2.5}(n+r)\kappa^2\log1/\epsilon)$.
\end{proof}
%
%Using Ambainis' algorithm instead of HHL algorithm, the complexity is $\widetilde O((n+\log(n+r))(r+5n)\kappa/\sqrt{1/(2n)}^3\log1/\epsilon)=\widetilde O((n^{2.5}(n+r)\kappa\log1/\epsilon)$.

The best classic probabilistic algorithm for 3-SAT was $1.334^n$ given in \cite{3sat-1}. In order for our quantum algorithm to perform better $n$ should be
$\ge 64$, if $\kappa$ is not too big.

\section{Solving Boolean quadratic equations and cryptanalysis}
\label{sec-ca}

Cryptanalysis of stream ciphers, block ciphers, certain hash functions, and MPKC can be reduced to the solving of Boolean multivariate  quadratic equations (BMQ). In this section, we will apply our quantum algorithm
to the analysis of these cryptosystems.
%
%$\FS=\{f_1,\ldots,f_r\}\subset\C[\X]$ is called a {\em multivariate quadratic polynomial system (MQ)} if $\deg(f_i)= 2$. Then, we have
%the following corollaries.
%%
%\begin{cor}
%For an MQ $\FS$, we have  $T_\FS=O(rn^2)$, $D\le 2n+2$ and the complexity is $\widetilde O(n^3r\kappa^2/\epsilon)$.
%\end{cor}

\subsection{Quantum algebraic attack against AES}
The Advanced Encryption Standard (AES), also known by its original name Rijndael, is a specification for the encryption of electronic data established by the U.S. National Institute of Standards and Technology in 2001 \cite{aes-or}.

Murphy and Robshaw \cite{aes} proposed a method to construct a Boolean equation system, solving of which consists of an algebraic attack against AES. We will use this approach to establish a BMQ.

Denote the 32-bit key length of AES as ${N_k}$ and the number of rounds as ${N_r}$.
Denote $p,\,c\in\F_2^{4N_k\times 8}$ as the plaintext and the ciphertext of AES, $w_0\in\F_2^{4N_k\times 8}$ as the key of AES,
$w_i\in\F_2^{4N_k\times 8}$ as the expanded key of AES,
$\bar w_i\in\F_2^{4N_k\times 8}$ as the image of $w_i$ under the S-box map in the key expansion step,
$x_i\in\F_2^{4N_k\times 8}$ as the state after the AddRoundKey step of AES,
and $y_i\in\F_2^{4N_k\times 8}$ as the state after the InvSubBytes step of AES, where $x_i(j,m)$ means the $m$-th bit at the $j$-th word of state $x$ for round $i$.
In the key expansion step, several states $\bar w_i$ are obtained as the image of $w_i$ under the S-box.
Then, an algebraic attack on the ${N_r}$-rounds AES with key length ${N_k}$
is to solve the following BMQ, denoted as AES-(${N_k},{N_r}$):
\begin{align*}
0&= x_0(j,m)+p(j,m)+w_0(j,m);\\
0&= x_i(j,m)+w_i(j,m)+\sum_{j',m'}
\alpha(j,m,j',m')y_{i-1}(j',m')&\text{ for }i=1,\cdots,N_r-1;\\
0&=c(j,m)+w_{N_r}(j,m)+y_{N_r-1}(5j\text{ mod }16,m);\\
\mathbf 0&=\mathcal S(x_i(j,0),\ldots,x_i(j,7),y_i(j,0),\ldots,y_i(j,7))&\text{ for }i=0,\cdots,N_r-1;\\
%\end{align*}
%
%\begin{align*}
%0&= x_0(j,m)+p(j,m)+w_0(j,m);\\
%0&= x_i(j,m)+w_i(j,m)+\sum_{j',m'}
%\alpha(j,m,j',m')y_{i-1}(j',m')&\text{ for }i=1,\cdots,N_r-1;\\
%0&=c(j,m)+w_{N_r}(j,m)+y_{N_r-1}(5j\text{ mod }16,m);\\
%\mathbf 0&=\mathcal S(x_i(j,0),\ldots,x_i(j,7),y_i(j,0),\ldots,y_i(j,7))&\text{ for }i=0,\cdots,N_r-1;\\
%
\mathbf 0&=\mathcal S(w_i(\bar j,0),\ldots,w_i(\bar j,7),\bar w_i(\bar j,0),\ldots,\bar w_i(\bar j,7))&\text{ for }\bar j=4N_k-4,\ldots,4N_k-1;\\
0&= w_i(\bar j,m)+w_{i-1}(\bar j,m)+\bar w_{i-1}(\bar j+13,m)+\chi(m,i)&\text{ for }\bar j=0,1,2;\\
0&= w_i(3,m)+w_{i-1}(3,m)+\bar w_{i-1}(12,m)+\chi(m,i).\\
&\text{For }N_k\le6:\\
0&= w_i(\bar j,m)+w_{i-1}(\bar j,m)+w_{i}(\bar j-4,m)&\text{ for }\bar j=4,\ldots,4N_k-1.\\
&\text{For }N_k>6:\\
\mathbf 0&=\mathcal S(w_i(\bar j,0),\ldots,w_i(\bar j,7),\bar w_i(\bar j,0),\ldots,\bar w_i(\bar j,7))&\text{ for }\bar j=12,\ldots,15;\\
0&= w_i(\bar j,m)+w_{i-1}(\bar j,m)+\bar w_{i}(\bar j-4,m)&\text{ for }\bar j=16,\ldots,19;\\
0&= w_i(\bar j,m)+w_{i-1}(\bar j,m)+w_{i}(\bar j-4,m)&\text{ for }\bar j=4,\ldots,15,20,\ldots,4N_k-1,
\end{align*}
where $j$ runs from $0$ to $(4N_k-1)$, and $m$ runs from $0$ to $7$. $\bar w_i(j,m)$, $x_i(j,m)$, and $y_i(j,m)$ are state variables, $w_i(j,m)$ are key variables,
$\mathcal S$ is a set of $39$ BMQ in $\FB_2[x_0,\ldots,x_7,y_0,\ldots,y_7]$
representing the Rijndael S-box, which can be found in the Appendix of this paper. $\chi$ is the round constant.
% can be represent by
%$39$ quadratic polynomials over $\FB_2[x_0,\ldots,x_7,y_0,\ldots,y_7]$.
%
Thus $x,y,w$ and $\bar w$ are Boolean indeterminates and other alphabets are known constants.
In the second group of equations, exactly $640$ of $\alpha(j,m,j',m')$ are $1$ for a given $i$.

The equation set $\mathcal S$ of the S-box is given in the Appendix (Section \ref{sec-aessb}),
which can be simplified as follows.
The original $\mathcal S$ is a BMQ with total sparseness $1688$.
By doing Gaussian elimination, we obtain a BMQ $\mathcal S$ with with total sparseness $1192$. We can introduce $1075$ new indeterminates $u_{ij}$ to split $\mathcal S$ into $3$-sparse BMQ. Thus $P(\mathcal S,3)$ consists of $1331$ quadratic binomials, $115$ quadratic polynomials, $989$ cubic polynomials, and $10$ quartic polynomials over $\C$.

Totally, the AES-(${N_k},{N_r}$) can be represented by a BMQ with number of indeterminates $n=96N_kN_r+32N_k+32N_r$ (if $N_k\le6$) or $n=96N_kN_r+32N_k+64N_r$ (if $N_k>6$), number of equations $r=220N_rN_k+64N_k+156N_r$ (if $N_k\le6$) or $r=220N_rN_k+64N_k+312N_r$ (if $N_k>6$), and total sparseness $T=4928N_rN_k+192N_k+5440N_r$ (if $N_k\le6$) or $T=4928N_rN_k+192N_k+10208N_r$ (if $N_k>6$).
By {Theorem \ref{cor-ec}}, we have
\begin{prop}
There is a quantum algorithm to obtain a solution
of   AES-$({N_k},{N_r})$
with complexity $\sqrt{2}(\log_2(5024N_kN_r+224N_k+5472N_r)+3)
(5024N_kN_r+224N_k+5472N_r)^{2.5}(34592N_kN_r+1376N_k+38112N_r)
c\kappa^2$ $\log_{2}1/\epsilon$ (if $N_k\le6$), or $\sqrt{2}(\log_2(5024N_kN_r+224N_k+10272N_r)+3)
(5024N_kN_r+224N_k+10272N_r)^{2.5}(34592N_kN_r+1376N_k+71520N_r)
c\kappa^2$ $\log_{2}1/\epsilon$ (if $N_k>6$)
  with probability $>1-\epsilon$, where $\kappa$ is the condition number of $\FS$ and $c$ is the complexity constant of the HHL algorithm.
\end{prop}

Set ${N_k}=4,6,8$, ${N_r}=4,6,8,10,12,14$, and $\epsilon=1\%$. We have the following complexities on quantum algebraic attack on various AESes.
From Table \ref{tab-1}, we can see that AES is secure under
quantum algebraic attack only if the condition number $\kappa$
is large.

\begin{table}[ht]\centering
\begin{tabular}{|c|c|c|c|c|c|c|}\hline
AES&${N_k}$&${N_r}$&\#Vars&\#Eqs&T-Sparseness& Complexity \\ \hline
AES-128&4&4&1792&4400&101376&$2^{68.61}c\kappa^2$\\ \hline
AES-128&4&6&2624&6472&151680&$2^{70.68}c\kappa^2$\\ \hline
AES-128&4&8&3456&8544&201984&$2^{72.16}c\kappa^2$ \\ \hline
AES-128&4&10&4288&10616&252288&$2^{73.30}c\kappa^2$ \\ \hline
AES-192&6&12&7488&18096&421248&$2^{76.59}c\kappa^2$ \\ \hline
AES-256&8&14&11904&29520&696384&$2^{78.53}c\kappa^2$ \\ \hline
\end{tabular}
\caption{Complexities of the quantum algebraic attack on AES}
\label{tab-1}
\end{table}

\subsection{Quantum algebraic attack against Trivium}
Trivium is a synchronous stream cipher designed by Canni\'ere and Preneel \cite{tri-or} in 2005 to provide a flexible trade-off between speed and gate count in hardware, and reasonably efficient software implementation, which has been specified as an International Standard under ISO/IEC 29192-3.
Trivium can be represented by the following
nonlinear feedback shift registers (NFSR) which
can also be considered as BMQ  \cite{tri} $\FS$:
%\begin{equation}\label{eq-tr}
{\small\begin{eqnarray}
&A(t + 93) =
A(t+24)+ C(t+45)+ C(t)+ C(t+1)C(t+2),& 0\le t\le N_r-67;\nonumber\\
&B(t + 84) =
B(t+6) + A(t+27) + A(t) + A(t+1)A_2(t+2),& 0\le t\le N_r-70;\label{eq-tr}\\
&C(t + 111) =
C(t+24) + B(t+15) + B(t) + B(t+1)B(t+2),& 0\le t\le N_r-67;\nonumber\\
&z(t) =
A(t+27)+ A(t) + B(t+15)+ B(t) + C(t+45) + C(t),& 0\le t\le N_r-1,\nonumber
\end{eqnarray}}
%\end{align}
%\end{equation}
where $A,B,C$ are state variables and $z$ is the output.
For an initial state $Z_0 = (A(0),\ldots,A(92) ,\\ B(0), \ldots, B(83) , C(0) ,\ldots,C(110))\in \F_2^{288}$,
we can generate the key sequence $z(0), z(1), \ldots, z(N_r$ $-1)$ with
the above NFSR.
Thus, for the $N_r$-round Trivium, $\FS$ consists of $(3N_r-201)$ quadratic polynomials of sparseness $5$, and $N_r$ linear polynomials of sparseness $7$, and with $(3N_r+87)$ indeterminates. Thus, $T=5(3N_r-201)+7N_r=22N_r-1005$.

The algebraic attach on the $N_r$-round Trivium is to solve
the BMQ \bref{eq-tr}, where $z(0),z(1),\ldots$, $z(N_r-1)$ are constants.
It is generally believed that for $N_r>288$,
 \bref{eq-tr} has a unique solution.

\begin{prop}
There is a quantum algorithm to find a solution for the $N_r$-round
Trivium equation system in time $2^{20.46}\log_2(N_r)N_r^{3.5}c\kappa^2\log1/\epsilon$ with probability $>1-\epsilon$, where $\kappa$ is the condition number of $\FS$ and $c$ is the complexity constant of the HHL algorithm.
\end{prop}
\begin{proof}
By {Theorem \ref{cor-ec}}, the complexity is $\sqrt{2}c(\log_2((3N_r+87)+(22N_r-1005))+3)
((3N_r+87)+(22N_r-1005))^{2.5}((3N_r+87)
+7(22N_r-1005))
\kappa^2\lceil\log_{2}1/\epsilon\rceil=\sqrt{2}c
\log_2(200N_r-7344)
(25N_r-918)^{2.5}(157N_r-6948)
\kappa^2\lceil\log_{2}1/\epsilon\rceil\le2^{20.46}\log_2(N_r)N_r^{3.5}c\kappa^2
\lceil\log_{2}1/\epsilon\rceil.$
\end{proof}

In Table \ref{tab-2}, we give the complexities for several $N_r$ assuming $\epsilon=1\%$.
From Table \ref{tab-2}, we can see that Trivium is secure under
quantum algebraic attack only if the condition number $\kappa$
is large.

\begin{table}[ht]\centering
\begin{tabular}{|c|c|c|c|c|c|}\hline
Round &\#Vars&\#Eqs&T-Sparseness& Complexity \\ \hline
288&951&951&5331&$2^{53.96}c\kappa^2$ \\ \hline
576&1815&2103&11667&$2^{57.94}c\kappa^2$ \\ \hline
1152&3543&4407&24339&$2^{61.71}c\kappa^2$ \\ \hline
2304&6999&9015&49683&$2^{65.38}c\kappa^2$ \\ \hline
\end{tabular}
\caption{Complexities of the quantum algebraic attack on Trivium}
\label{tab-2}
\end{table}

\subsection{Quantum algebraic attack against Keccak}
Keccak \cite{kec}, the winner of SHA-3 contest,  is the latest member of the Secure Hash Algorithm family of standards, released by NIST on August 5, 2015.
For Keccak-[$N_h,b,N_r$], we denote $N_h,b$, and $N_r$ as the output bit size, the state bit size, and the number of rounds.
Let $A_0(x,y,z)$ be the message, $A_i(x,y,z)$ be the state variable after applying the $\tau$ function for $i$-times,
and $B_i(x,y,z)$ be the state variable after applying the $\pi$ function for $i$-times, where $x,y\in\Z/5\Z$, $z\in\Z/w\Z$, and $w=b/25=1,2,4,\ldots,64$ is the bit length of each word. Then for Keccak-[$N_h,b,N_r$], we have the following BMQ $\FS$ \cite{wf}:
{\small
\begin{align*}
&B_{i}(3y+x,x,z-r(3y+x,x))=A_{i-1}(x,y,z)+\sum_{j=0}^4A_{i-1}(x-1,j,z)+\sum_{j=0}^4A_{i-1}(x+1,j,z-1);\\
&A_{i}(x,y,z)=B_{i}(x,y,z)+(1-B_{i}(x+1,y,z))B_{i}(x+2,y,z),\text{ for }x\ne0\text{ or }y\ne0;\\
&A_{i}(0,0,z)=B_{i}(0,0,z)+(1-B_{i}(1,0,z))B_{i}(2,0,z)+RC(z)
\end{align*}}
for $i=1,\ldots,N_r$, $x,y=0,\ldots,4$, $z=0,\ldots,w$.
In the preimage attack on Keccak, $r(3y+x,x)$ and $RC(z)$ are known constants, the first $N_b$ of $A_{N_r}(x,y,z)$ are the known Hash output, and $A_i(x,y,z)\ (i<N_r)$ and $B_i(x,y,z)$ are indeterminates. Thus we have $n=2bN_r$ indeterminates and $r=(2b-1)N_r+N_h$ Boolean quadratic equations with total sparseness $T=401N_rw+101N_h/25-101w$.

\begin{prop}
For the BMQ Keccak-[$N_h,b,N_r$], there is a quantum algorithm to find a preimage in time $2\sqrt{2} c\log_2(401N_rw+26N_h/25)
(401N_rw+26N_h/25)^{2.5}(3609N_rw+384N_h/25)
\kappa^2$ $\lceil\log_{2}1/\epsilon\rceil\le2^{18.83}\log_2(N_rb)N _r^{3.5}b^{3.5}c\kappa^2$ $\lceil\log_{2}1/\epsilon\rceil$ with probability $>1-\epsilon$, where $\kappa$ is the condition number of $\FS$ and $c$ is the complexity constant of the HHL algorithm.
\end{prop}
\begin{proof}
We have $n=2bN_r$, $r=(2b-1)N_r+N_h$, and $T=401N_rw+101N_h/25-101w$.
By {Theorem \ref{cor-ec}}, the complexity is
$\sqrt{2}c(\log_2(2bN_r+(401N_rw+101N_h/25-101w))+3)
(2bN_r+(401N_rw+101N_h/25-101w))^{2.5}
(2bN_r+7(401N_rw+101N_h/25-101w))
\kappa^2\lceil\log_{2}1/\epsilon\rceil\le
\sqrt{2}c\log_2(3608N_rw+808N_h/25)
(451N_rw+101N_h/25)^{2.5}(2857N_rw+707N_h/25)
\kappa^2\lceil\log_{2}1/\epsilon\rceil\le2^{18.83}\log_2(N_rb)N _r^{3.5}b^{3.5}$ $c\kappa^2\lceil\log_{2}1/\epsilon\rceil$.
\end{proof}

Setting $N_h=224,256,384,512$, $N_r=24$, $b=1600$ and $\epsilon=1\%$, the
complexities for various $(N_h, b, N_r)$ are given in Table \ref{tab-3}.
From Table \ref{tab-3}, we can see that Keccak is secure under
quantum algebraic attack only if the condition number $\kappa$
is large.

\begin{table}[ht]\centering
\begin{tabular}{|c|c|c|c|c|c|c|}\hline
$N_h$&$b$&$N_r$&\#Vars&\#Eqs&T-Sparseness& Complexity \\ \hline
224&1600&24&76800&77000&610377&$2^{78.25}c\kappa^2$ \\ \hline
256&1600&24&76800&77032&610506&$2^{78.25}c\kappa^2$ \\ \hline
384&1600&24&76800&77160&611023&$2^{78.25}c\kappa^2$ \\ \hline
512&1600&24&76800&77288&611540&$2^{78.25}c\kappa^2$ \\ \hline
\end{tabular}
\caption{Complexities of the quantum preimage attack on Keccak}
\label{tab-3}
\end{table}

The best known traditional attacks on Keccak were given in \cite{slg}
and \cite{xwang1}.
In \cite{slg}, practical collision attacks against the $5$-round Keccak-224
and an instance of the 6-round Keccak collision challenge were given.
In \cite{xwang1}, key recovery attacks were given for $4$- to $8$-round Keccak.

\subsection{Quantum algebraic attack against MPKC}

Multivariate Public Key Cryptosystem (MPKC) is one of the candidates
for post-quantum cryptography \cite{dingjt}.
An MPKC is generally constructed as follows
\begin{equation}
H = L\circ G \circ R = (h_1(\X),\ldots,h_r(\X))
\end{equation}
where $L\in GL(m,\F_2)$, $R\in GL(n,\F_2)$, and $G: \F_2^n \rightarrow \F_2^m$
is a quadratic map whose inversion can be efficiently computed.
$L$ and $R$ are the secret keys and $H$ is the public map.
The direct algebraic attack against the MPKC is to solve the BMQ:
\begin{equation}\label{eq-mpkx1}
y_1=h_1(\X),\ldots,y_r=h_r(\X)
\end{equation}
where $\X = (x_1,\ldots,x_n)$ is the plaintext and
$\Y=(y_1,\ldots,y_r)$ is the known ciphertext.
Note that the BMQ in \bref{eq-mpkx1} are dense. We have
\begin{prop}\label{p-mq1}
For dense BMQ $\FS=\{f_1,\ldots,f_r\}\subset\mathcal R_2[\X]$, there is a quantum algorithm to obtain a solution in time $\sqrt{2}c(\log_2(n+(n^2+3n-4)r/2)+3)(n+(n^2+3n-4)r/2)^{2.5}
(n+7(n^2+3n-4)r/2)\kappa^2\lceil\log_{2}1/\epsilon\rceil = \widetilde O(n^7r^{3.5}\kappa^2\log1/\epsilon)$
  with probability $>1-\epsilon$, where $\kappa$ is the condition number of $\FS$ and $c$ is the complexity constant of the HHL algorithm.
\end{prop}
\begin{proof}
If $\FS$ is a dense BMQ, $T=(n+1)(n+2)r/2$. By {Theorem \ref{cor-ec}}, we can find a solution of $\FS=0$ in time
$\sqrt{2}c(\log_2(n+(n^2+3n-4)r/2)+3)(n+(n^2+3n-4)r/2)^{2.5}
(n+7(n^2+3n-4)r/2)\kappa^2\lceil\log_{2}1/\epsilon\rceil = \widetilde O(n^7r^{3.5}\kappa^2\log1/\epsilon)$.
\end{proof}

%With Corollary \ref{p-mq1}, we have a {\em quantum algebraic
%attack algorithm} for MPKC, which is polynomial in
%$n$, $m$, and the condition number of \bref{eq-mpkx1}.

\begin{cor}\label{p-mq2}
Suppose $r=\gamma n$. Then there is a quantum algebraic attack against
MPKC in time $\widetilde O(\gamma n^{10.5}\kappa^2\log1/\epsilon)$  with probability $>1-\epsilon$, where $\kappa$ is the condition number of \bref{eq-mpkx1}.
\end{cor}

Related to this problem, the BMQ Challenge is to solve a given
random BMQ with $ m = 2n$ or $n=1.5m$ over the finite fields
$\F_2$, $\F_{2^8}$ \cite{mqc}.
Considering the field $\FB_{2^8}=\FB_2[\alpha]/(\alpha^8+\alpha^4+\alpha^3+\alpha+1)$, each variable $x$ over $\FB_{2^8}$ is a sum of eight Boolean variables, that is $x=\sum_{i=0}^7x_i\alpha^i$. Then $\FS=\{f_1,\ldots,f_r\}\subset\FB_{2^8}[x_1,\ldots,x_n]$ can be rewritten as $\FS_1=\{f_{11},\ldots,f_{18},f_{21},\ldots,f_{m8}\}\subset\mathcal R_{2}[x_{11},\ldots,x_{18},x_{21},\ldots,x_{n8}]$, where $x_{ij}$ and $f_{ij}$ denote the $j$-th bit of $x_i$ and $f_i$.
As a consequence, equation solving over $\FB_{2^8}$ has the same complexity
as Boolean equation solving.
For the BMQ challenge \cite{mqc}, $m=2n$ or $n=1.5m$ implies the complexity is $\widetilde O(T^{3.5}\kappa^2\log1/\epsilon)<\widetilde O(n^{10.5}\kappa^2\log1/\epsilon)$.

The best known deterministic algebraic algorithms to solve
the BMQ are the Gr\"obner basis method \cite{gb-ff}
which has complexity $O(2^{0.841n})$ under certain regularity
condition for the equation system,
and the multiplication free characteristic set method \cite{cs-ff}
which has bit complexity $O(2^n)$ for general BMQ.
Although exponential in $n$,  these methods
had been used to solve BMQ from cryptanalysis with $n=128$.

\begin{rem}
From the above discussion, we can see that AES, Trivium, Keccak, and MPKC are secure under quantum algebraic attack only if the condition numbers of the related Boolean equation systems are large.
This suggests that a possible {\em new quantum criterion for cryptosystem design}: the Boolean equation system of the cryptosystem has a large condition number.
\end{rem}

\section{Conclusion}
We give two quantum algorithms  to find the Boolean
solutions of a polynomial system in $\C[\X]$
and to solve Boolean equations in ${\mathcal R}_2[\X]$ in any given probability,
whose complexities are polynomial in the number of variables,
the total sparseness of the equation system, and the condition number of the equation system.
As a consequence, we achieved exponential speedup for sparse
Boolean equation solving if the condition number of the equation system is small.

The main idea of the algorithm is  first reducing the problem
of Boolean equation solving to the problem of
finding the Boolean solutions of a polynomial system over $\C$
and then solving the Macaulay linear system
of the polynomial system over $\C$ with the modified HHL algorithm to obtain the Boolean solutions based on the properties of quantum states.

The new quantum algorithm is used to give quantum algebraic attack
against  major cryptosystems AES, Trivium, and SHA-3/Keccak and show that these ciphers are secure under quantum algebraic attack only if the condition numbers
of their equation systems are large.
Similar results hold for MPKC, which is a candidate for post-quantum
cryptosystems.

We also use the quantum algorithms to three famous problems
from computational theory: the 3-SAT problem, the graph isomorphism
problem, and the subset sum problem and show that the
complexities to solve these problems are polynomial
in the input size and the condition number of their
corresponding equation system.

One of the major problems for future study is on the condition number: either to estimate the condition number of the cryptosystems and the 3-SAT problem, or to find new quantum method to solve Boolean systems, which has less relation with the condition number.
It is also interesting to extend the method proposed in this paper
to more general equation systems, such as equation solving over
the finite fields \cite{qa-opt} or the field of complex numbers.

\newpage
\section{Appendix. Equations for the AES S-Box}
\label{sec-aessb}

We list the 39 Boolean quadratic polynomials for the AES S-Box used in this paper.

$x_5x_7+x_5x_6+x_3x_7+x_3x_6+x_2x_4+x_1x_7+
x_1x_6+x_1x_5+x_1x_3+x_1x_2+x_0x_7+x_0x_3+
x_0x_2+x_6y_7+x_7y_6+x_6y_6+x_7y_5+x_5y_5+
x_7y_4+x_1y_4+x_2y_3+x_0y_3+x_6y_2+x_4y_2+
x_3y_2+x_0y_2+x_4y_0+x_2y_0+x_7+x_5+x_3+y_7+
y_2+y_0+1$,

$x_6x_7+x_5x_7+x_4x_7+x_4x_6+
x_4x_5+x_3x_4+x_2x_5+x_1x_7+x_1x_6+x_1x_5+
x_1x_4+x_1x_3+x_1x_2+x_0x_5+x_0x_1+x_6y_6+
y_5y_7+x_3y_4+y_4y_7+y_4y_5+x_5y_3+x_0y_3+
y_3y_6+y_3y_4+x_3y_2+x_0y_2+y_2y_4+y_2y_3+
x_5y_0+x_3y_0+x_1y_0+y_0y_7+y_0y_3+y_0y_1+
x_5+x_3+x_0+y_2+1$,

$x_1y_7+x_0y_7+y_6y_7+
x_7y_5+x_6y_5+y_5y_7+x_7y_4+x_5y_3+x_2y_3+
y_3y_6+x_2y_2+x_0y_2+y_2y_5+x_6y_1+x_4y_1+
x_1y_1+y_1y_2+x_6y_0+x_5y_0+x_4y_0+y_0y_7+
y_0y_6+y_0y_5+y_0y_4+y_0y_3+x_3+x_1+y_3+y_2+
y_1+1$,

$x_6x_7+x_4x_6+x_3x_7+x_2x_7+x_1x_4+
x_0x_6+x_0x_3+x_6y_7+x_4y_7+x_3y_7+x_7y_6+
x_3y_6+x_7y_5+x_7y_4+x_1y_4+x_5y_3+x_4y_3+
x_1y_3+x_6y_2+x_2y_2+x_6y_1+x_5y_1+x_3y_1+
x_1y_1+x_0y_1+x_7y_0+x_6y_0+x_5y_0+x_3y_0+
x_2y_0+x_1y_0+x_6+x_1$,

$x_6y_7+x_5y_7+x_1y_7+
x_0y_7+x_5y_6+x_4y_6+x_3y_6+x_4y_4+x_3y_4+
x_2y_4+x_4y_3+x_3y_3+y_3y_5+x_2y_2+y_2y_7+
y_2y_4+y_2y_3+x_7y_1+x_4y_1+x_3y_1+x_1y_1+
y_1y_7+y_1y_6+y_1y_5+y_1y_2+x_7y_0+x_5y_0+
x_4y_0+x_3y_0+y_0y_4+y_0y_3+y_0y_2+x_6+x_4+
x_3+y_2$,

$x_2y_7+y_5y_6+x_1y_4+x_7y_3+x_2y_3+
x_1y_3+x_0y_3+y_3y_6+y_3y_4+x_4y_2+x_2y_2+
x_0y_2+y_2y_6+y_2y_5+y_2y_3+x_5y_1+x_3y_1+
y_1y_7+y_1y_5+y_1y_4+y_1y_3+y_1y_2+x_5y_0+
x_4y_0+x_3y_0+x_0y_0+y_0y_6+y_0y_1+x_7+x_6+
x_3+x_2+x_1+y_4+y_2+y_1+y_0$,

$x_5y_7+x_3y_6+
x_2y_6+x_0y_6+x_6y_5+x_5y_5+x_0y_5+x_6y_4+
x_5y_4+x_3y_4+x_2y_4+x_0y_4+y_4y_7+y_3y_6+
y_3y_5+x_3y_2+x_2y_2+x_0y_2+y_2y_7+y_2y_6+
y_2y_5+x_3y_1+x_2y_1+y_1y_7+y_1y_6+y_1y_3+
x_5y_0+x_4y_0+x_1y_0+y_0y_7+y_0y_2+x_6+x_1+
y_4+y_3+y_1$,

$x_6y_7+x_3y_7+x_0y_7+x_2y_6+
x_4y_4+x_2y_4+x_0y_4+x_7y_3+x_6y_3+x_3y_3+
x_2y_3+x_1y_3+x_0y_3+x_5y_2+x_2y_2+x_1y_2+
x_3y_1+x_2y_1+x_1y_1+x_3y_0+x_6+x_2+x_1+x_0+
y_2$,

$x_7y_7+x_4y_7+x_1y_7+x_7y_6+x_6y_6+
x_1y_6+y_6y_7+x_7y_5+x_6y_5+x_2y_5+x_0y_5+
x_6y_4+x_4y_4+x_2y_4+y_4y_5+x_4y_3+x_3y_3+
x_2y_3+x_1y_3+x_0y_3+y_3y_6+y_3y_5+x_3y_2+
y_2y_4+x_6y_1+x_5y_1+x_4y_1+y_1y_5+y_1y_2+
x_6y_0+x_2y_0+x_1y_0+y_0y_6+y_4+y_3$,

$x_4y_7+
x_3y_7+x_4y_6+x_2y_6+x_1y_6+x_0y_4+x_3y_3+
x_1y_3+x_0y_3+x_7y_2+x_3y_2+x_2y_2+x_1y_2+
x_0y_2+x_7y_1+x_6y_1+x_5y_1+x_4y_1+x_3y_1+
x_0y_1+x_4+x_2+x_1+y_2$,

$x_3x_6+x_2x_5+
x_1x_4+x_1x_2+x_0x_4+x_0x_1+x_4y_6+x_2y_6+
x_1y_6+x_7y_5+x_7y_4+x_2y_4+x_6y_3+x_4y_3+
x_3y_3+x_2y_3+x_7y_2+x_0y_2+x_4y_1+x_3y_1+
x_5y_0+x_2y_0+x_1y_0+x_4+x_1+y_4+y_3+y_0+
1$,

$x_4x_7+x_2x_3+x_1x_5+x_1x_4+x_0x_7+
x_0x_6+x_0x_5+x_0x_4+x_0x_1+x_7y_7+x_4y_7+
x_1y_6+x_2y_4+x_1y_4+x_0y_4+x_7y_3+x_4y_3+
x_3y_3+x_4y_2+x_3y_2+x_0y_2+x_7y_1+x_5y_1+
x_2y_1+x_1y_1+x_0y_1+x_1y_0+x_0y_0+x_7+x_1+
x_0+y_2$,

$x_6x_7+x_4x_5+x_3x_7+x_3x_5+x_2x_5+
x_2x_4+x_2x_3+x_1x_7+x_0x_6+x_0x_4+x_2y_7+
x_0y_7+x_1y_6+x_6y_5+x_2y_4+x_6y_3+x_5y_3+
x_2y_3+x_6y_2+x_4y_2+x_3y_2+x_2y_2+x_1y_2+
x_7y_1+x_5y_1+x_6y_0+x_5y_0+x_4y_0+x_3y_0+
x_0y_0$,

$x_5x_7+x_3x_6+x_1x_7+x_1x_2+x_0x_4+
x_0x_3+x_1y_7+x_2y_6+x_1y_6+x_6y_5+x_4y_5+
x_2y_5+x_0y_4+x_5y_3+x_2y_3+x_6y_2+x_5y_2+
x_1y_2+x_0y_2+x_7y_1+x_6y_0+x_5y_0+x_0y_0+
x_6+x_1+y_6+y_2+y_1$,

$x_1y_7+x_0y_7+x_2y_6+
x_6y_5+x_2y_5+x_4y_4+x_3y_4+x_2y_4+x_1y_4+
x_0y_4+x_5y_3+x_2y_3+x_1y_3+x_7y_2+x_4y_2+
x_3y_2+x_1y_2+x_6y_1+x_3y_1+x_2y_1+x_1y_1+
x_0y_1+x_5y_0+x_0y_0+x_3+y_2$,

$x_5x_7+x_3x_6+
x_1x_7+x_1x_2+x_0x_4+x_0x_3+x_6y_7+x_3y_7+
x_0y_7+x_4y_6+x_5y_5+x_2y_5+x_4y_4+x_3y_3+
x_1y_2+x_0y_2+x_7y_1+x_5y_1+x_4y_1+x_2y_1+
x_7y_0+x_6y_0+x_1y_0+x_6+x_5+x_4+x_2+x_0+
y_7$,

$x_7y_7+x_7y_6+x_6y_6+x_4y_6+x_2y_6+
x_1y_6+x_7y_5+x_2y_5+x_7y_4+x_1y_4+x_0y_4+
y_4y_6+x_7y_3+x_3y_3+y_3y_5+y_3y_4+x_7y_2+
x_4y_2+x_3y_2+y_2y_4+y_2y_3+y_1y_6+x_3y_0+
x_2y_0+x_1y_0+y_0y_6+y_0y_4+y_0y_2+y_0y_1+
x_5+x_4+x_3+x_1+y_4+y_2+y_1$,

$x_7y_6+y_6y_7+
x_7y_5+x_3y_5+y_5y_7+x_7y_4+x_0y_4+y_4y_7+
y_4y_5+x_6y_3+x_4y_3+x_3y_3+x_1y_3+y_3y_7+
y_3y_6+y_3y_5+y_3y_4+x_6y_2+x_5y_2+x_4y_2+
x_1y_2+x_0y_2+y_2y_5+x_6y_1+x_5y_1+x_4y_1+
x_0y_1+y_1y_7+y_1y_4+y_1y_2+x_5y_0+x_4y_0+
y_0y_7+y_0y_5+y_0y_4+x_0+y_0$,

$x_6y_6+x_1y_6+
x_5y_5+x_3y_4+x_0y_4+x_7y_3+x_6y_3+x_5y_3+
x_1y_3+x_0y_3+x_7y_2+x_7y_1+x_6y_1+x_4y_1+
x_3y_1+x_0y_1+x_4y_0+x_2y_0+x_7+x_6+x_4+x_3+
x_0+y_2$,

$x_6x_7+x_5x_7+x_4x_6+x_3x_7+x_2x_7+
x_2x_5+x_1x_7+x_0x_6+x_0x_1+x_7y_7+x_3y_7+
x_1y_7+x_5y_6+x_0y_4+x_6y_3+x_1y_3+x_7y_2+
x_6y_2+x_5y_2+x_4y_2+x_3y_2+x_4y_1+x_2y_1+
x_5y_0+x_3y_0+x_5+x_2+x_1+x_0+y_3+
y_2$,

$x_0y_7+x_2y_6+x_0y_6+x_5y_5+x_0y_5+
x_6y_4+x_3y_4+x_2y_4+x_0y_4+x_7y_3+x_6y_3+
x_4y_3+x_3y_3+x_2y_3+x_4y_2+x_3y_2+x_4y_1+
x_2y_1+x_7y_0+x_2y_0+x_0y_0+x_7+x_4+x_2+x_1+
y_4+y_1+y_0$,

$x_5x_7+x_5x_6+x_4x_6+x_3x_5+
x_2x_6+x_2x_5+x_1x_7+x_1x_6+x_0x_7+x_0x_6+
x_0x_3+x_0x_1+x_6y_7+x_3y_6+x_3y_4+x_0y_4+
x_7y_3+x_4y_3+x_3y_3+x_5y_2+x_4y_2+x_6y_1+
x_4y_1+x_2y_1+x_1y_1+x_1y_0+x_5+x_0+y_4+
y_2$,

$x_5x_7+x_5x_6+x_3x_7+x_3x_4+x_2x_6+
x_2x_4+x_1x_4+x_1x_3+x_1x_2+x_0x_6+x_7y_7+
x_6y_7+x_4y_7+x_3y_7+x_1y_6+x_1y_4+x_7y_3+
x_3y_3+x_2y_3+x_6y_2+x_2y_2+x_0y_2+x_7y_1+
x_5y_1+x_4y_1+x_3y_1+x_6y_0+x_3y_0+x_1y_0+
x_6+x_3+x_2+y_4$,

$x_5x_7+x_4x_6+x_4x_5+
x_3x_5+x_2x_7+x_2x_4+x_2x_3+x_0x_4+x_0x_1+
x_4y_6+x_2y_6+x_7y_5+x_7y_4+x_6y_4+x_1y_4+
x_6y_3+x_5y_3+x_2y_3+x_1y_3+x_7y_2+x_5y_2+
x_7y_1+x_6y_1+x_0y_1+x_4y_0+x_3y_0+x_1y_0+
x_3+x_1+y_4+y_3$,

$x_5x_7+x_4x_6+x_3x_7+
x_3x_5+x_3x_4+x_2x_4+x_1x_6+x_1x_3+x_1x_2+
x_0x_7+x_4y_7+x_3y_7+x_0y_7+x_4y_6+x_2y_6+
x_0y_6+x_5y_5+x_0y_5+x_6y_4+x_6y_3+x_6y_2+
x_4y_2+x_3y_2+x_1y_2+x_5y_1+x_4y_1+x_5y_0+
x_1y_0+x_0y_0+x_7+x_5+y_1+y_0$,

$x_6x_7+
x_4x_6+x_4x_5+x_2x_6+x_2x_5+x_2x_3+x_1x_7+
x_1x_5+x_1x_4+x_1x_3+x_0x_2+x_1y_7+x_2y_6+
x_1y_6+x_6y_5+x_2y_5+x_2y_4+x_3y_3+x_7y_2+
x_3y_2+x_2y_2+x_7y_1+x_5y_1+x_1y_1+x_7+x_3+
x_0+y_6+y_4+y_2+y_1$,

$x_7y_5+x_7y_4+x_2y_4+
x_7y_3+x_6y_3+x_0y_3+x_5y_2+x_4y_2+x_0y_2+
x_7y_1+x_3y_1+x_2y_1+x_0y_1+x_7y_0+x_5y_0+
x_4y_0+x_2y_0+x_1y_0+x_2+x_1+x_0+y_4+y_2+y_1+
1$,

$x_5x_6+x_3x_7+x_3x_6+x_2x_5+x_2x_4+
x_1x_6+x_1x_5+x_1x_4+x_1x_3+x_1x_2+x_0x_7+
x_0x_2+x_0x_1+x_5y_7+x_1y_6+x_4y_4+x_2y_4+
x_0y_4+x_5y_3+x_4y_3+x_1y_3+x_0y_2+x_5y_1+
x_2y_1+x_1y_1+x_3y_0+x_0y_0+x_6+x_5+x_1+y_3+
y_2$,

$x_6y_7+x_3y_7+x_1y_7+x_6y_6+x_5y_6+
x_1y_6+x_5y_5+x_1y_5+x_6y_4+x_6y_3+x_5y_3+
x_4y_3+x_3y_2+x_7y_1+x_6y_1+x_1y_1+x_6y_0+
x_3y_0+x_2y_0+x_5+y_7+y_3+y_2+y_1+
y_0$,

$x_5x_6+x_4x_6+x_3x_6+x_3x_5+x_3x_4+
x_1x_7+x_1x_5+x_0x_3+x_0x_2+x_5y_7+x_1y_6+
x_6y_5+x_2y_5+x_2y_4+x_0y_4+x_6y_2+x_4y_2+
x_2y_2+x_1y_2+x_6y_1+x_5y_1+x_4y_1+x_2y_1+
x_5y_0+x_3y_0+x_7+x_6+x_4+x_0+y_6+y_4+y_2+
1$,

$x_5y_7+x_4y_7+x_3y_7+x_1y_7+x_0y_7+
x_7y_6+x_4y_6+x_7y_5+x_3y_5+x_7y_4+x_4y_4+
x_3y_4+x_7y_3+x_6y_3+x_4y_3+x_3y_3+x_1y_3+
x_6y_2+x_5y_2+x_2y_2+x_2y_1+x_7y_0+x_3y_0+
x_2y_0+x_1y_0+y_6+y_2+y_0$,

$x_7y_7+x_4y_7+
x_1y_7+x_6y_6+x_5y_6+x_5y_5+x_2y_3+x_1y_3+
x_0y_3+x_5y_2+x_4y_2+x_2y_2+x_1y_2+x_7y_1+
x_0y_1+x_6y_0+x_5y_0+x_4y_0+x_0y_0+x_6+x_2+
x_1+x_0+y_7$,

$x_7y_7+x_4y_7+x_2y_7+x_4y_6+
x_6y_5+x_5y_5+x_1y_5+x_6y_4+x_4y_4+x_7y_3+
x_6y_3+x_6y_2+x_5y_2+x_0y_2+x_6y_1+x_2y_1+
x_6y_0+x_4y_0+x_5+x_4+x_3+x_1+x_0+y_7+y_4+
y_2$,

$x_4x_6+x_4x_5+x_3x_5+x_3x_4+x_2x_7+
x_2x_6+x_2x_4+x_2x_3+x_1x_6+x_1x_5+x_1x_4+
x_1x_2+x_0x_7+x_0x_6+x_0x_2+x_0x_1+x_0y_5+
x_6y_4+x_3y_4+x_0y_4+x_6y_3+x_3y_3+x_2y_3+
x_6y_2+x_5y_2+x_4y_2+x_0y_2+x_7y_1+x_5y_1+
x_0y_0+x_7+y_1+1$,

$x_6y_7+y_6y_7+x_7y_5+
y_5y_7+x_7y_4+y_4y_7+x_7y_3+x_6y_3+x_2y_3+
x_1y_3+y_3y_5+y_3y_4+x_7y_2+x_6y_2+x_0y_2+
y_2y_3+x_4y_1+x_3y_1+x_2y_1+x_1y_1+x_0y_1+
y_1y_5+y_1y_2+x_7y_0+x_1y_0+y_0y_7+y_0y_6+
y_0y_3+y_0y_1+x_2+y_5+y_4$,

$x_4y_7+x_3y_7+
x_1y_7+x_5y_6+x_4y_6+x_1y_6+x_4y_5+x_0y_5+
x_7y_4+x_6y_4+x_4y_4+x_2y_4+x_1y_4+x_5y_3+
x_4y_3+x_3y_3+x_0y_3+x_3y_2+x_0y_1+x_6y_0+
x_1y_0+x_6+x_4+x_2+y_6+y_2+y_1$,

$x_5x_7+
x_4x_7+x_2x_5+x_2x_3+x_1x_7+x_1x_5+x_0x_7+
x_0x_6+x_0x_5+x_0x_4+x_0x_3+x_4y_5+x_6y_4+
x_5y_3+x_4y_3+x_0y_3+x_7y_2+x_3y_2+x_2y_2+
x_5y_1+x_3y_1+x_0y_1+x_5y_0+x_3y_0+x_0y_0+
x_7+x_5+x_1+x_0+y_4+y_0+1$,

$x_7y_7+x_4y_7+
x_0y_7+x_6y_6+x_5y_6+x_0y_6+x_0y_5+x_6y_4+
x_5y_4+x_3y_4+x_7y_3+x_6y_3+x_0y_3+x_7y_2+
x_3y_2+x_2y_2+x_6y_1+x_3y_1+x_1y_1+x_5y_0+
x_4y_0+x_3y_0+x_1y_0+x_3+y_3$,

$x_4x_5+x_3x_7+
x_3x_6+x_3x_4+x_2x_7+x_2x_5+x_2x_3+x_1x_6+
x_1x_4+x_1x_3+x_0x_7+x_1y_7+x_2y_6+x_1y_6+
x_6y_5+x_2y_5+x_4y_4+x_3y_4+x_7y_3+x_3y_2+
x_5y_1+x_2y_1+x_0y_1+x_4y_0+x_3y_0+x_0y_0+
x_3+y_3+y_2+y_1+y_0$.


\begin{thebibliography}{99}

\bibitem{hhl-ha1}
  Aharonov, D. and  Ta-Shma, A.,
  Adiabatic quantum state generation and statistical zero knowledge,
  {\em Proc. STOC'03}, 20-29, ACM Press, New York,   2003.

\bibitem{hhl-new}
 Ambainis, A.,
 Variable time amplitude amplification and a faster quantum algorithm for solving systems of linear equations,
 {\em Proc. STACS}, 636-647, 2012.

\bibitem{gb-ff}
 Bardet, M., Faug\`ere, J.C., Salvy, B., Spaenlehauer, P.J.,
 On the complexity of solving quadratic Boolean systems,
 {\em Journal of Complexity},
 29(1), 53-75, 2013.

\bibitem{ha}
 Berry, D.W., Ahokas, G., Cleve, R., Sanders, B.C., Efficient quantum algorithms for simulating sparse Hamiltonians.
 {\em Communications in Mathematical Physics}, 270(2): 359-371, 2007.

\bibitem{he}
 Berry, D.W., Childs, A.M., Kothari, R.,
 Hamiltonian simulation with nearly optimal dependence on all parameters,
 {\em Proc. 56th FOCS}, 792-809, 2015.

\bibitem{kec}
 Bertoni, G., Daemen, J., Peeters, M., Van Assche, G.,
 Keccak sponge function family main document.
 Submission to NIST (Round 2) 3 (2009): 30.

\bibitem{soldeg-1}
 Caminata, A. and Gorla, E.,
 Solving multivariate polynomial systems and an invariant from commutative
 algebra, arXiv1706.06319, 2017.

\bibitem{tri-or}
 Canni\'ere, C.D., Preneel, B.,
 Trivium,
 in {\em New Stream Cipher Designs: The eSTREAM Finalists},
 LNCS, vol. 4986, 244-266, Springer, 2008.
%In: Robshaw, M.J.B., Billet, O., eds.

\bibitem{hhl-ci}
 Cao, Y., Daskin, A., Frankel, S., Kais S.,
 Quantum circuit design for solving linear systems of equations,
 {\em Journal of Molecular Physics}, 1675-1680, 2012.

\bibitem{qa-opt}
 Chen, Y.A.,  Gao, X.S.,   Yuan, C.M.,
 Quantum Algorithms for Optimization and Polynomial Systems Solving over Finite Fields,
 arxiv 1802.03856, 2018.


\bibitem{childs}
  Childs, A.M.,
  Quantum algorithms: equation solving by simulation,
  {\em Nature Physics},  5(12), 861-861, 2009.

\bibitem{hhl-ha2}
  Childs, A.M.,  Cleve, R.,  Deotto, E.,   Farhi, E.,
   Gutmann, S.,   Spielman, D.A.,
  Exponential algorithmic speedup by a quantum walk,
  {\em Proc. STOC'03},  59-68, ACM Press, New York,   2003.

\bibitem{hhl-ha3}
 Childs, A.M.,  Kothari, R.,  Somma, R.D.,
 Quantum Algorithm for Systems of Linear Equations with
 Exponentially Improved Dependence on Precision,
 {\em SIAM J. Comput.}, 46(6), 1920-1950, 2017.

\bibitem{xl}
 Courtois, N., Klimov, A., Patarin, J., Shamir, A.,
 Efficient algorithms for solving overdefined
 systems of multivariate polynomial equations,
 {\em Eurocrypt'00}, LNCS, vol. 1807, 392-407, Springer, 2000.

\bibitem{cox2}
Cox, D., Little, J., O'Shea, D.,
{\em Using Algebraic Geometry},
Springer, 1998.

\bibitem{aes-or}
 Daemen, J. and Rijmen, V.,
 {\em AES Proposal: Rijndael}, NIST, 1999.

%\bibitem{dr}
% Daemen, J. and Rijmen, V.,
% {\em The design of Rijndael},
% Springer-Verlag, 2002.


\bibitem{dingjt}
 Ding, J., Gower, J.E., Schmidt, D.S.,
 {\em Multivariate Public Key Cryptosystems},
 Springer, 2006.

\bibitem{f4}
 Faugere, J.C.,
 A new efficient algorithm for computing Gr{\"o}bner bases (F4),
 {\em J. Pure Appl. Algebra},
 139, 61-88, 1999.

\bibitem{qmq2}
 Faug\`ere, J.C.,  Horan, K.,  Kahrobaei, D.,
 Kaplan, M.,  Kashefi, E. L Perret, L.,
 Fast quantum algorithm for solving multivariate quadratic equations,
 arXiv:1712.07211, 2017.

\bibitem{cs-ff}
 Gao, X.S. and  Huang, Z.,
 Characteristic set algorithms for equation solving in finite fields,
 {\em Journal of Symbolic Computation}, 47, 655-679, 2012.



%\bibitem{gls}
% Guo, J., Liu, M., Song, L.,
% Linear structures: Applications to cryptanalysis of round-reduced Keccak,
% {\em ASIACRYPT 2016}, Part I, pp. 249-274. Springer Berlin Heidelberg, 2016.

\bibitem{aesq1}
  Grassl, M.,  Langenberg, B.,  Roetteler,  M.,  Steinwandt, R.,
 Applying Grover¡¯s algorithm to AES: quantum resource estimates
 {\em International Workshop on Post-Quantum Cryptography
 Post-Quantum Cryptography}, 29-43, Springer, 2016.


\bibitem{grover}
  Grover, L.K.,
 A fast quantum mechanical algorithm for database search,
 {\em Proc. STOC'96}, 212-219, ACM Press, 1996.


\bibitem{hhl}
 Harrow, A.W., Hassidim, A., Lloyd, S.,
 Quantum algorithm for linear systems of equations.
 {\em Physical Review Letters}, 103(15): 150502, 2009.

\bibitem{xwang1}
  Huang, S.,  Wang, X., Xu, G., Wang, M., Zhao, J.,
  Conditional cube attack on reduced-round Keccak sponge function,
  {\em EUROCRYPT 2017},   259-288, Springer, 2017.

\bibitem{hsl}
 Huang, Z., Sun Y., Lin D.,
 On the efficiency of solving boolean polynomial systems with the characteristic set method, arXiv:1405.4596, 2016.

\bibitem{laz83}
 Lazard, D.,
 Gr\"obner bases, Gaussian elimination and resolution of systems of algebraic equations,
 {\em Proc. Eurocal 83},
 LNCS, vol. 162, 146-156, Springer, Berlin, 1983.

\bibitem{mac}
 Macaulay, F.S.,
 Some formulas in elimination,
 {\em Proc. of the London Mathematical Society}, 35(1), 3-38, 1902.

\bibitem{aes}
 Murphy, S. and Robshaw, M.,
 Essential algebraic structure within the AES.
 {\em CRYPTO'02}, 1-16, 2002.

\bibitem{qmq1}
 Schwabe, P. and  Westerbaan, B.,
 Solving binary MQ with Grover¡¯s algorithm,
 {\em SPACE'16}, LNCS 10076, 303-322, 2016.

\bibitem{3sat-1}
 Sch\"oning, U.,
 A probabilistic algorithm for k-SAT and constraint satisfaction problems,
{\em Proc. FOCS'99}, 410-414, 1999.

\bibitem{c19}
 Shewchuk, J.R.,
 An introduction to the conjugate gradient method without the agonizingpain, Tech. Rep. CMU-CS-94-125, Carnegie Mellon University,
 Pittsburgh, Pennsylvania, 1994.

\bibitem{sha}
 Shannon, C.E.,
 A mathematical theory of communication,
 {\em Bell Syst. Tech. J.}, 27(3), 379-423 and 623-656, 1948.

\bibitem{shor}
 Shor, P.W.,
 Polynomial-Time Algorithms for Prime Factorization and
 Discrete Logarithms on a Quantum Computer,
 {\em SIAM J. Comput.}, 26(5), 1484-1509, 1997.

\bibitem{slg}
 Song, L., Liao, G., Guo, J.,
 Non-full Sbox linearization: applications to collision attacks on round-reduced Keccak.
 {\em  CRYPTO'17},   428-451, Springer, 2017.

\bibitem{tri}
 Teo, S.G., Wong, K.K.H., Bartlett, H., Simpson L, Dawson, E.,
 Algebraic analysis of Trivium-like ciphers.
 Proceedings of the Twelfth Australasian Information Security Conference-Volume 149. Australian Computer Society, Inc., 77-81, 2014.

\bibitem{wf}
 Wu, C.K. and Feng, D.,
 {\em Boolean Functions and their Applications in Cryptography},
 Springer, 2016.

\bibitem{mqc}
 Yasuda, T., Dahan, X., Huang, Y.J., Takagi, T., Sakurai, K.,
 MQ Challenge: hardness evaluation of solving multivariate quadratic problems,
 {\em The NIST Workshop on Cybersecurity in a Post-Quantum World},
 2015, https://www.mqchallenge.org/.

\end{thebibliography}
\end{document}